\documentclass[a4paper,UKenglish,cleveref, autoref, thm-restate]{lipics-v2021}
\pdfoutput=1 %uncomment to ensure pdflatex processing (mandatatory e.g. to submit to arXiv)
\hideLIPIcs  %uncomment to remove references to LIPIcs series (logo, DOI, ...), e.g. when preparing a pre-final version to be uploaded to arXiv or another public repository

\usepackage{booktabs}
\usepackage{commath}
\usepackage{mathtools}
\usepackage{siunitx}
\usepackage{timetravel}
\usepackage[textsize=tiny]{todonotes}
\usepackage{tabularx}
\usepackage{multirow}
\setCounters{theorem}

\DeclareMathOperator{\dist}{dist}

\newcommand{\niceAngle}{\varphi}

\makeatletter
\DeclareRobustCommand{\bfseries}{%
  \not@math@alphabet\bfseries\mathbf
  \fontseries\bfdefault\selectfont
  \boldmath
}
\makeatother

\title{Efficient Shortest Paths in Scale-Free Networks with Underlying
  Hyperbolic Geometry}

\titlerunning{Shortest Paths in Networks with Underlying Hyperbolic
  Geometry}

\author{Thomas Bläsius}{Karlsruhe Institute of Technology, Karlsruhe, Germany}{thomas.blaesius@kit.edu}{https://orcid.org/0000-0003-2450-744X}{} %mandatory, please use full name; only 1 author per \author macro; first two parameters are mandatory, other parameters can be empty. Please provide at least the name of the affiliation and the country. The full address is optional. Use additional curly braces to indicate the correct name splitting when the last name consists of multiple name parts.
\author{Cedric Freiberger}{Hasso Plattner Institute, University of Potsdam, Potsdam, Germany}{cedric.freiberger@student.hpi.de}{}{} %mandatory, please use full name; only 1 author per \author macro; first two parameters are mandatory, other parameters can be empty. Please provide at least the name of the affiliation and the country. The full address is optional. Use additional curly braces to indicate the correct name splitting when the last name consists of multiple name parts.
\author{Tobias Friedrich}{Hasso Plattner Institute, University of Potsdam, Potsdam, Germany}{tobias.friedrich@hpi.de}{https://orcid.org/0000-0003-0076-6308}{} %mandatory, please use full name; only 1 author per \author macro; first two parameters are mandatory, other parameters can be empty. Please provide at least the name of the affiliation and the country. The full address is optional. Use additional curly braces to indicate the correct name splitting when the last name consists of multiple name parts.
\author{Maximilian Katzmann}{Karlsruhe Institute of Technology, Karlsruhe, Germany}{maximilian.katzmann@kit.edu}{https://orcid.org/0000-0002-9302-5527}{} %mandatory, please use full name; only 1 author per \author macro; first two parameters are mandatory, other parameters can be empty. Please provide at least the name of the affiliation and the country. The full address is optional. Use additional curly braces to indicate the correct name splitting when the last name consists of multiple name parts.
\author{Felix Montenegro-Retana}{Hasso Plattner Institute, University of Potsdam, Potsdam, Germany}{felix.montenegro-retana@student.hpi.de}{}{} %mandatory, please use full name; only 1 author per \author macro; first two parameters are mandatory, other parameters can be empty. Please provide at least the name of the affiliation and the country. The full address is optional. Use additional curly braces to indicate the correct name splitting when the last name consists of multiple name parts.
\author{Marianne Thieffry}{Hasso Plattner Institute, University of Potsdam, Potsdam, Germany}{marianne.thieffry@student.hpi.de}{}{} %mandatory, please use full name; only 1 author per \author macro; first two parameters are mandatory, other parameters can be empty. Please provide at least the name of the affiliation and the country. The full address is optional. Use additional curly braces to indicate the correct name splitting when the last name consists of multiple name parts.

\authorrunning{Bläsius et al.} %mandatory. First: Use abbreviated first/middle names. Second (only in severe cases): Use first author plus 'et al.'

\Copyright{Thomas Bläsius, Cedric Freiberger, Tobias Friedrich, Maximilian Katzmann, Felix Montenegro-Retana and Marianne Thieffry} %mandatory, please use full first names. LIPIcs license is "CC-BY";  http://creativecommons.org/licenses/by/3.0/


 %mandatory; please add comma-separated list of keywords

\category{} %optional, e.g. invited paper

\relatedversion{A preliminary version of the paper appeared in~\cite{bff-espsf-18}.} %optional, e.g. full version hosted on arXiv, HAL, or other respository/website
\relatedversiondetails{Preliminary version}{http://doi.org/10.4230/LIPIcs.ICALP.2018.20} %linktext and cite are optional

\nolinenumbers %uncomment to disable line numbering

\EventEditors{John Q. Open and Joan R. Access}
\EventNoEds{2}
\EventLongTitle{42nd Conference on Very Important Topics (CVIT 2016)}
\EventShortTitle{CVIT 2016}
\EventAcronym{CVIT}
\EventYear{2016}
\EventDate{December 24--27, 2016}
\EventLocation{Little Whinging, United Kingdom}
\EventLogo{}
\SeriesVolume{42}
\ArticleNo{23}
\begin{document}

\maketitle

% REQUIRED
\begin{abstract}
  A standard approach to accelerating shortest path algorithms on
  networks is the bidirectional search, which explores the graph from
  the start and the destination, simultaneously.  In practice this
  strategy performs particularly well on scale-free real-world
  networks. Such networks typically have a heterogeneous degree
  distribution (e.g., a power-law distribution) and high clustering
  (i.e., vertices with a common neighbor are likely to be connected
  themselves). These two properties can be obtained by assuming an
  underlying hyperbolic geometry.

  To explain the observed behavior of the bidirectional search, we
  analyze its running time on hyperbolic random graphs and prove that
  it is $\mathcal {\tilde O}(n^{2 - 1/\alpha} + n^{1/(2\alpha)} + \delta_{\max})$
  with high probability, where $\alpha \in (1/2, 1)$ controls the power-law
  exponent of the degree distribution, and $\delta_{\max}$ is the maximum
  degree. This bound is sublinear, improving the obvious worst-case
  linear bound. Although our analysis depends on the underlying
  geometry, the algorithm itself is oblivious to it.
\end{abstract}

\section{Introduction}
\label{sec:introduction}

% Shortest pat is fundamental and important
One of the most fundamental graph problems consists of finding a
shortest path between two vertices in a network.  Besides being of
independent interest, many algorithms use shortest path queries as a
subroutine.
% BFS is optimal solution but still to slow
On unweighted graphs, such queries can be answered in linear time
using a \emph{breadth-first search (BFS)}.  Though this is optimal in
the worst case, it is not efficient enough when dealing with large
networks or problems involving many shortest path queries.

% Speedup using heuristic: bidirectional BFS
A way to heuristically improve the run time, is to use a
\emph{bidirectional BFS}~\cite{p-bhspp-69}.  It runs two searches,
simultaneously exploring the graph from the start and the destination.
The shortest path is found once the two search spaces touch.
% bidirektionale Suche in Praxis super wichtig
Being one of the standard heuristics, the bidirectional BFS is widely
used in practice (e.g., in route planning).
% Observations in practice: homogeneous networks: factor 2
On homogeneous networks (where most vertices have similar degrees,
like road networks) this typically leads to a speedup factor of about
two.
% Heterogeneous networks: asymptotic improvement.
However, on heterogeneous networks (having many vertices of low degree
and only few vertices of very high degree, like social networks or the
internet) experiments indicate that the bidirectional BFS yields an
asymptotic running time improvement~\cite{bn-kaabra-16}.

% Theoretic results
Despite being such a fundamental heuristic, theory completely fails
its main purpose of predicting and explaining the observed behavior.
The theoretical worst case running time overshoots the observations by
a lot.  A more promising approach is the average-case analysis by
Borassi and Natale~\cite{bn-kaabra-16}, which considers instances that
are drawn from certain probability distributions instead of assuming
the worst case.  Their results are summarized in the first row of
Table~\ref{tab:resultComparison}.  The analysis covers a variety of
random graph models.  On the one hand these include homogeneous
networks where the degree distribution has bounded variance,
e.g. Erd\H{o}s-R\'{e}nyi random graphs.  On the other hand, they also
consider heterogeneous networks where the variance of the degree
distribution is unbounded, e.g. Chung-Lu random graphs with power-law
exponent $\beta \in (2, 3)$.  However, the results, again, do not
match what is observed in practice, as it predicts shorter running
times on homogeneous networks than on heterogeneous ones.

\begin{table}[t]
  \centering
  \begin{tabular}{c c c} 
    & \textbf{Homogeneous} & \textbf{Heterogeneous} \\
    \toprule
    \begin{tabular}{@{}c@{}}\textbf{Independent} \\ \textbf{Edges}\end{tabular} & \begin{tabular}{@{}c@{}}Bounded \\ Variance \\[0.6em] $m^{\frac{1}{2} + o(1)}$~\cite{bn-kaabra-16} \end{tabular}  & \begin{tabular}{@{}c@{}}Unbounded \\ Variance \\[0.6em] $m^{\frac{4-\beta}{2} + o(1)}$~\cite{bn-kaabra-16} \end{tabular} \\
    \midrule
    \begin{tabular}{@{}c@{}}\textbf{Underlying} \\ \textbf{Geometry}\end{tabular} & \begin{tabular}{@{}c@{}}Euclidean \\ Random Graphs \\[0.6em] $\Theta(n)$ \footnotesize{(Folklore)} \end{tabular} & \begin{tabular}{@{}c@{}}Hyperbolic \\ Random Graphs \\[0.3em] $\tilde{\mathcal{O}}(n^{2\frac{\beta - 2}{\beta - 1}} + n^{\frac{1}{\beta - 1}})$ \footnotesize{(This paper)} \end{tabular} \\
    \bottomrule
  \end{tabular} 
  \caption{Probabilistic bounds on the running time of the
    bidirectional BFS obtained by analyzing different random graph
    models.  The considered models (and associated results) are
    arranged by the heterogeneity of the corresponding degree
    distributions of the graphs and the (in)dependence of edges.
    Here, $n$ and $m$ denote the number of vertices and edges in the
    graph, respectively.  The parameter $\beta \in (2, 3)$ denotes the
    power-law exponent of the degree distribution in the considered
    heterogeneous networks.  }
  \label{tab:resultComparison}
\end{table}

% Problem of average case: independence.
The fundamental obstacle that prevents the average-case analysis from
producing convincing explanations is that the considered random graph
models are not realistic.  They assume that edges in the graph are
independent of each other.  However, real-world networks typically
exhibit \emph{locality}, i.e., edges in an evolving network tend to
form between vertices that are already close in the network.

% We model dependencies using geometry.
We resolve this discrepancy by modeling edge dependencies using
geometry and extend the comparison in Table~\ref{tab:resultComparison}
by adding the second row.
% Geometric random graphs.
Generally, \emph{geometric random graphs} are obtained by randomly
distributing vertices in some metric space (e.g. the Euclidean plane)
and connecting any two vertices with a probability that depends on
their distance.
% Hyperbolisch fuer heterogene Graphen.
In this framework, heterogeneous networks (i.e., networks on which the
bidirectional BFS has been observed to perform particularly well) can
be obtained by using the hyperbolic plane as the underlying geometry.

% Results.
In this paper, we analyze the bidirectional BFS on random graph models
with an underlying geometry.  We prove that, with high probability,
the bidirectional BFS has a sublinear worst-case running time on the
heterogeneous networks generated by the hyperbolic random graph model.
Additionally, it is not hard to see why there is no asymptotic speedup
on the homogeneous networks generated by the Euclidean random graph
model.  Both results match previous empirical observations.  Finally,
we interpret these insights and discuss how the heterogeneity of the
degree distribution and an underlying geometry affect the running time
of the bidirectional breadth first search.

\paragraph*{Related Work}
\label{sec:related-work}

The research on scale-free networks has gained a lot of attention for
quite some time now.  Therefore, it is no surprise that the
extensively studied problem of computing shortest paths has also been
considered in the context of such graphs~\cite{ask-spqcn-12,
  phzs-fafap-12, lfc-dsspa-17}.  However, the bidirectional search
that was introduced in 1969~\cite{p-bhspp-69} and that has since
become one of the standard search heuristics, has only recently been
examined on scale-free networks.  In fact, there are only two
theoretical explanations for the performance improvements obtained
using this heuristic, both using an average-case analysis that
considers one or more random graph models~\cite{lr-bspagacb-89,
  bn-kaabra-16}.

A model that yields a better representation of real-world networks
than the ones considered before, is the hyperbolic random graph model
introduced by Krioukov et al.~\cite{kpk-h-10}.  The generated graphs
feature a heterogeneous degree distribution, high clustering, and a
small diameter; properties that are often observed in real-world
networks.  These properties emerge naturally from the hyperbolic
geometry.  Moreover, the model is conceptually simple, which makes it
accessible to mathematical analysis.  For these reasons it has gained
popularity in different research areas and has been studied from
different perspectives.

From the network-science perspective, the goal is to gather knowledge
about real-world networks.  This is, for example, achieved by assuming
that a real-world network has a hidden underlying hyperbolic geometry,
which can be revealed by embedding it into the hyperbolic
plane~\cite{ama-m-16, bpk-sihm-10}.

From the mathematical perspective, the focus lies on studying
structural properties.  The degree distribution and
clustering~\cite{gpp-rhg-12}, diameter~\cite{fk-dhrg-18, ms-k-19},
component structure~\cite{bfm-gcrhg-13, km-slcrhg-19}, clique
size~\cite{bfk-chrg-17}, and separation properties~\cite{bfk-hrg-16}
have been studied successfully.

Additionally, there is the algorithmic perspective, which is the focus
of this paper.  Usually algorithms are analyzed by proving worst-case
running times.  Though this is the strongest possible performance
guarantee, it is rather pessimistic as practical instances rarely
resemble worst-case instances.  Techniques leading to a more realistic
analysis include parameterized or average-case complexity.  The latter
is based on the assumption that instances are drawn from a certain
probability distribution.
% Thus, its explanatory power depends on how realistic the
% distribution is.
For hyperbolic random graphs, the maximum clique, as well as the
minimum vertex cover can be computed in polynomial
time~\cite{bfk-chrg-17, bffk-svcpt-20}, and there are several
algorithmic results based on the fact that hyperbolic random graphs
have sublinear tree width~\cite{bfk-hrg-16}.  Moreover, there is a
compression algorithm that can store a hyperbolic random graph using
$\mathcal{O}(n)$ bits in expectation~\cite{bkl-sgirglt-17,p-rgmcs-14}.
Finally, a close approximation of the shortest path between two
vertices can be found using greedy routing, which visits only
$\mathcal{O}(\log \log n)$ vertices for most start--destination
pairs~\cite{bkl-graswp-17}.  The downside of most of these algorithms
is that they need to know the underlying geometry, i.e., the
coordinates of each vertex, which is a rather unrealistic assumption
for real-world networks.  In contrast to that, we analyze an algorithm
that is oblivious to the underlying geometry.

\paragraph*{Outline}

After a brief introduction to hyperbolic random graphs in
Section~\ref{sec:preliminaries}, we examine the bidirectional BFS in
Section~\ref{sec:bbbfs}.  We start by briefly arguing why the
bidirectional BFS gives no asymptotic speedup over the standard BFS on
Euclidean random graphs in Section~\ref{sec:bidirectional-euclidean}.
Afterwards, in Section~\ref{sec:bidir-search-hyperb} we rigorously
analyze the bidirectional BFS on hyperbolic random graphs.
Section~\ref{sec:conc-bounds-sum-deg} contains concentration bounds
that were left out in Section~\ref{sec:bbbfs} to improve readability.
In Section~\ref{sec:conclusion}, we conclude by comparing our
theoretical results to empirical data and interpret them.

\section{Preliminaries}
\label{sec:preliminaries}

Let $G = (V, E)$ be an undirected and unweighted graph.  We denote the
number of vertices and edges with $n$ and $m$, respectively.  The
\emph{neighborhood} of a vertex $v \in V$ is $N(v) = \{ w \in
V~\vert~\{v, w\} \in E \}$.  The \emph{degree} of $v$ is $\deg(v) =
|N(v)|$.  We denote the maximum degree with $\delta_{\max}$.  The soft
$\mathcal{O}$-notation $\mathcal {\tilde O}$ suppresses
poly-logarithmic factors in $n$.

\subsection{The Hyperbolic Plane}

The major difference between hyperbolic and Euclidean geometry is the
exponential expansion of space.  In the hyperbolic plane, a circle of
radius $r$ has area $2\pi(\cosh(r)-1)$ and circumference $2\pi
\sinh(r)$, with $\cosh(x) = (e^x+e^{-x})/2$ and $\sinh(x) =
(e^x-e^{-x})/2$, both growing as $e^x/2 \pm o(1)$.
To identify points, we use polar coordinates with respect to a
designated origin $O$ and a ray starting at $O$.  A point $p$ is
uniquely determined by its \emph{radius}~$r$, which is the distance to
$O$, and the \emph{angle} (or \emph{angular coordinate}) $\niceAngle$
between the reference ray and the line through $p$ and~$O$.
In illustrations, we use the \emph{native representation}, obtained by
interpreting the hyperbolic coordinates as polar coordinates in the
Euclidean plane; see Figure~\ref{fig:native}~(left).  Due to the
exponential expansion, line segments bend towards the origin~$O$.
\begin{figure}
  \centering
  \includegraphics[scale=0.975]{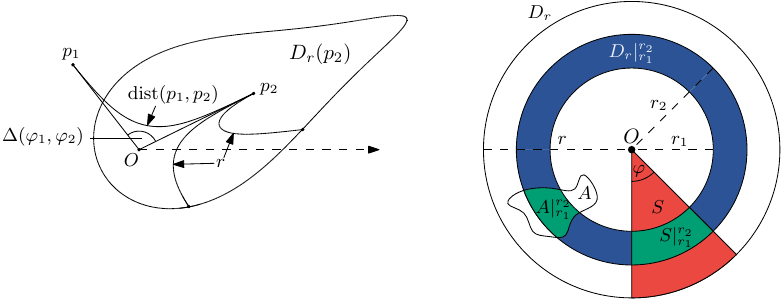}
  \caption{Left: Points and several line segments in the native
    representation of the hyperbolic plane. A disk of radius $r$ is
    centered at~$p_2$.  Right: Geometric shapes and their
    intersections.  Sector $S$ has an angular width of $\varphi$.}
  \label{fig:native}
\end{figure}
Let $p_1=(r_1, \niceAngle_1)$ and $p_2=(r_2, \niceAngle_2)$ be two
points.  The \emph{angular distance} between $p_1$ and $p_2$ is the
angle between the rays from the origin through $p_1$ and $p_2$.
Formally, it is $\Delta(\niceAngle_1, \niceAngle_2) = \pi - |\pi -
|\niceAngle_1 - \niceAngle_2||$.  The hyperbolic distance $\dist(p_1,
p_2)$ is given by
\begin{align*}
  \cosh(\dist(p_1, p_2)) = \cosh(r_1)\cosh(r_2) - \sinh(r_1)\sinh(r_2)\cos(\Delta(\niceAngle_1, \niceAngle_2)).
\end{align*}
Note how the angular coordinates make simple definitions cumbersome as
angles are considered modulo $2\pi$, leading to a case distinction
depending on where the reference ray lies.  Whenever possible, we
implicitly assume that the reference ray was chosen such that we do
not have to compute modulo $2\pi$.  Thus, the above angular distance
between $p_1$ and $p_2$ simplifies to~$|\niceAngle_1 - \niceAngle_2|$.
A third point $p = (r, \niceAngle)$ \emph{lies between} $p_1$
and~$p_2$ if
$\Delta(\niceAngle, \niceAngle_1) + \Delta(\niceAngle, \niceAngle_2) =
\Delta(\niceAngle_1, \niceAngle_2)$.

Throughout the paper, we regularly use different geometric shapes that
are mostly based on disks centered at the origin $O$, as can be seen
in Figure~\ref{fig:native} (right).  With $D_{r}(p)$ we denote the
disk of radius $r$ around a point $p$, i.e., the set of points that
have distance $r$ to $p$.  For disks, that are centered in the origin
$O$, we simplify the notation and set $D_r \coloneqq D_r(O)$. The
restriction of a disk $D_r$ to all points with angular coordinates in
a certain interval is called \emph{sector}, which we usually denote
with the letter $S$.  Its \emph{angular width} is the length of this
interval.  For an arbitrary set of points~$A$, we use $A|_{r_1}^{r_2}$
to denote the restriction of $A$ to points with radii in $[r_1, r_2]$,
i.e., $A|_{r_1}^{r_2} = A \cap (D_{r_2} \setminus D_{r_1})$.

\subsection{Hyperbolic Random Graphs}

A \emph{hyperbolic random graph} is generated by drawing $n$ points
uniformly at random in a disk of the hyperbolic plane and connecting
pairs of points whose distance is below a threshold.  More precisely,
the model depends on two parameters~$C$ and~$\alpha$ that are assumed
to be constants.  The generated graphs have a power-law degree
distribution with power-law exponent $\beta = 2\alpha + 1$ and a
constant average degree depending on $C$.  The parameter $\alpha$ is
assumed to be in the range $(1/2, 1)$, yielding power-law exponents
$\beta \in (2, 3)$.  Exponents outside of this range are atypical for
hyperbolic random graphs. For $\beta < 2$ the average degree of the
generated networks diverges, while for $\beta > 3$ the graphs
decompose into small components (of size sublinear in $n$) and the
variance of the degree distribution is no longer unbounded.  In
contrast, it is unbounded for $\beta \in (2, 3)$, resulting in very
heterogeneous degree distributions.  Moreover, in this range the
obtained networks have a giant component of size
$\Omega(n)$~\cite{bfm-lchmcn-15}, and all other components have at
most polylogarithmic size with high probability~\cite[Corollary
13]{fk-dhrg-18}.  Note that a bidirectional BFS could completely
explore a non-giant component in $\tilde{\mathcal{O}}(1)$ time and
either return the shortest path (if both vertices are in the same
non-giant component) or conclude that the vertices are in different
components.  Therefore, we only consider the case when the two
considered vertices are both in the giant component in the remainder
of the paper.

When generating a hyperbolic random graph, the $n$ points are sampled
within the disk $D_R$ of radius $R = 2\log n + C$.  For each vertex,
the angular coordinate is drawn uniformly from $[0, 2\pi]$.  Its
radius $r$ is sampled according to the probability density function
$f(r)$, which can then be used to define the joint distribution of
angles and radii $f(r, \varphi)$.  They are given by
\begin{align}
\label{eq:probDensity}
  f(r) = \frac{\alpha \sinh(\alpha r)}{\cosh(\alpha R) - 1} = \Theta(e^{-\alpha(R - r)}) ~\quad\text{and}~\quad f(r, \varphi) = \frac{1}{2\pi} f(r)
\end{align}
for $r \in [0, R]$.  For $r > R$, $f(r) = f(r, \varphi) = 0$.  Two
vertices are connected by an edge if and only if their hyperbolic
distance is at most~$R$.  The above probability distribution is a
natural choice as the probability for a vertex ending up in a certain
region is proportional to its area (at least for $\alpha = 1$).  Note
that the exponential growth in $r$ reflects the fact that the area of
a disk grows exponentially with the radius.  It follows that a
hyperbolic random graph has few vertices with high degree close to the
center of the disk and many vertices with low degree near its
boundary.  The following lemma is common knowledge; for the sake of
completeness we give a short proof.

\begin{lemma}
  \label{lem:smallerRadiusIncreasesNeighborhood}
  Let $G$ be a hyperbolic random graph. Furthermore, let $v_1, v_2$ be
  two vertices with radii $r_1 \le r_2 \le R$, respectively, and with the
  same angular coordinate. Then $N(v_2) \subseteq N(v_1)$.
\end{lemma}
\begin{proof}
  Let $w \in N(v_2)$, i.e., $\dist(v_2, w) \le R$. Now consider the
  triangle $v_2Ow$, which is completely contained in the disk of
  radius $R$ around $w$ (since $\dist(v_2, w) \le R $ and $r(w) \le
  R$). Since disks are convex and $v_1$ lies on the line from $O$ to
  $v_2$, it is part of the triangle and therefore also contained in
  this disk.  Consequently, $\dist(v_1, w) \le R$ and thus $w \in N(v_1)$.
\end{proof}

Given two vertices with fixed radii $r_1$ and $r_2$, their hyperbolic
distance grows with increasing angular distance.  The maximum angular
distance such that they are still
adjacent~\cite[Lemma~3.1]{gpp-rhg-12}~is
\begin{align}
  \label{eq:maxAdjacentAngle}
  \theta(r_1, r_2) &= \arccos\left( \frac{\cosh(r_1)\cosh(r_2) - \cosh(R)}{\sinh(r_1)\sinh(r_2)} \right) \notag \\
              &= 2e^{\frac{R - r_1 - r_2}{2}}(1 + \Theta(e^{R - r_1 - r_2})),
\end{align}
assuming $r_2 \ge R - r_1$.  Otherwise, we have $r_1 + r_2 < R$,
meaning two vertices with these radii are adjacent, independent of
their angular distance.

The probability that a sampled vertex falls into a given subset $A
\subseteq D_R$ of the disk is given by its probability measure $\mu(A)
= \iint_Af(r, \varphi) \dif \varphi \dif r$, which can be thought of
as the area of $A$.  There are two types of regions we encounter
regularly: disks $D_r$ with radius $r$ centered at the origin and
disks $D_R(r, \varphi)$ of radius $R$ centered at a point $(r,
\varphi)$.  Note that the measure of $D_R(r, \varphi)$ gives the
probability that a random vertex lies in the neighborhood of a vertex
with position $(r, \varphi)$.  Gugelmann et
al.~\cite[Lemma~3.2]{gpp-rhg-12} showed that
\begin{align}
  \label{eq:originBallApproximation}
  &\mu(D_r) = e^{-\alpha(R - r)}(1 + o(1)), \text{ and}\\
  \label{eq:expectedDegree}
  &\mu(D_R(r, \varphi)) = \Theta(e^{-r/2}).
\end{align}

For a given region $A \subseteq D_R$ of the disk, let $X_1, \dots,
X_n$ be random variables with $X_i = 1$ if vertex $i$ lies in $A$ and
$X_i = 0$ otherwise.  Then $X = \sum_{i=1}^n X_i$ is the number of
vertices lying in $A$.  By the linearity of expectation, we obtain
that the expected number of vertices in $A$ is $\mathbb E[X] = \sum_{i
  = 1}^{n} \mathbb E[X_i] = n \mu(A)$.

Often, determining the expected value of a random variable is not
sufficient to obtain meaningful statements.  Therefore, we
additionally classify events depending on how likely they are to
occur.  We say that an event holds \emph{with high probability}, if it
occurs with probability $1 - \mathcal{O}(1/n)$.  Moreover, we say that
an event holds \emph{asymptotically almost surely} if it occurs with
probability $1 - o(1)$.

To show that certain random variables are concentrated around their
expectation (i.e., with high probability the outcome does not deviate
much from the expected value) we regularly use the following
Chernoff-Hoeffding bound.
% 
% \begin{theorem}[{Chernoff-Hoeffding Bound~\cite[Theorem
%   1.1]{dp-cmara-12}}]
%   \label{thm:chernoff-bound-original}
%   Let $X_1, \ldots, X_n$ be $n$ independent random variables with $X_i \in
%   \{0, 1\}$ and let $X$ be their sum.  If $t > 2e\mathbb{E}[X]$, then
%   \begin{align*}
%     \Pr[X > t] \le 2^{-t}.
%   \end{align*}
% \end{theorem}

% \begin{corollary}
%   \label{col:chernoff-bound}
%   Let $X_1, \ldots, X_n$ be $n$ independent random variables with $X_i \in
%   \{0, 1\}$ and let $X$ be their sum.  If $f(n) = \Omega(\log n)$ is an
%   upper bound for $\mathbb{E}[X]$, then for any constant $c$ there is
%   a constant $c'$ such that $X \le c'f(n)$ holds with probability $1 -
%   \mathcal{O}(n^{-c})$.
% \end{corollary}

\begin{theorem}[{Chernoff Bound~\cite[Theorem 1.1]{dp-cmara-12}}]
  \label{thm:chernoff-bound-original}
  Let $X_1, \dots, X_n$ be independent random variables with $X_i \in
  \{0, 1\}$ and let $X$ be their sum.  Then, 
  \begin{align*}
    \Pr[X > t] &\le 2^{-t} &&\text{for $t > 2e\mathbb{E}[X]$~and}\\
    \Pr[X < (1 - \varepsilon)\mathbb{E}[X]] &\le e^{- \varepsilon^2/2 \cdot \mathbb{E}[X]} &&\text{for $\varepsilon \in (0, 1)$}.
  \end{align*}
\end{theorem}

Usually, it suffices to show that a random variable does not exceed a
certain upper bound or drop below a lower bound with high probability.
The following corollaries show that sufficiently large upper and lower
bounds on the expected value suffice to obtain concentration.

\begin{corollary}
  \label{col:chernoff-bound}
  Let $X_1, \dots, X_n$ be independent random variables with $X_i \in
  \{0, 1\}$ and let $X$ be their sum. Further, let $f(n) =
  \Omega(\log(n))$ be such that $\mathbb{E}[X] \le f(n)$ and let $c$
  be a constant.  Then, $X = \mathcal{O}(f(n))$ holds with probability
  $1 - O(n^{-c})$.
\end{corollary}
\begin{proof}
  We prove the statement by showing that the probability for the
  complementary event (i.e., $X$ is more than a constant factor larger
  than $f(n)$) is $\mathcal{O}(n^{-c})$ for any $c$.  Since
  $\mathbb{E}[X] \le f(n)$, we can choose a constant $c_1$
  sufficiently large such that $c_1 f(n) > 2e \mathbb{E}[X]$.  Thus,
  by Theorem~\ref{thm:chernoff-bound-original} it holds that
  \begin{align*}
    \Pr[X > c_1f(n)] \le 2^{-c_1f(n)}.
  \end{align*}
  Moreover, we have $f(n) = \Omega(\log n)$.  Consequently, there
  exists another constant $c_2$ such that $f(n) \ge c_2\log n$ for
  sufficiently large $n$.  We obtain
  \begin{align*}
    \Pr[X > c_1f(n)] \le 2^{-c_1 c_2 \log n} \le n^{-c_1 c_2}
  \end{align*}
  for $n$ sufficiently large.  Finally, we can chose $c_1$ such that
  $c_1 > c/c_2$, which yields the claim.
\end{proof}

\begin{corollary}
  \label{col:chernoff-lower-bound}
  Let $X_1, \dots, X_n$ be independent random variables with $X_i \in
  \{0, 1\}$ and let $X$ be their sum. Further, let $f(n) = \omega(\log
  n)$ be such that $f(n) \le \mathbb{E}[X]$ and let $c$ be a constant.
  Then, $X \in \Omega(f(n))$ holds with probability $1 - O(n^{-c})$.
\end{corollary}
\begin{proof}
  Analogous to the proof of Corollary~\ref{col:chernoff-bound} we
  prove the statement by showing that the probability for the
  complementary event (i.e., $X$ is more than a constant factor
  smaller than $f(n)$) is $\mathcal{O}(n^{-c})$ for any $c$.  Let
  $\varepsilon$ be a constant with $\varepsilon \in (0, 1)$.  The
  following inequalities are obtained by first using the fact that
  $f(n) \le \mathbb{E}[X]$, applying the second statement of
  Theorem~\ref{thm:chernoff-bound-original}, again applying $f(n) \le
  \mathbb{E}[X]$, and finally using $f(n) \in \omega(\log n)$:
  \begin{align*}
    \Pr[X < (1 - \varepsilon)f(n)] &\le \Pr[X < (1-\varepsilon)\mathbb{E}[X]] \\
                                   &\le e^{-\varepsilon^2/2 \cdot \mathbb{E}[X]} \\
                                   &\le e^{-\varepsilon^2/2 \cdot f(n)} \\
                                   &= e^{-\varepsilon^2/2 \cdot \omega(\log n)} \\
                                   &= n^{-\omega(1)}.
  \end{align*}
\end{proof}

Finally, the following lemma shows that statements about the
neighborhood of a vertex with fixed angular coordinate can be extended
to hold for arbitrary angular coordinates, with a small penalty in
certainty.

\begin{figure}[t]
  \centering
  \includegraphics{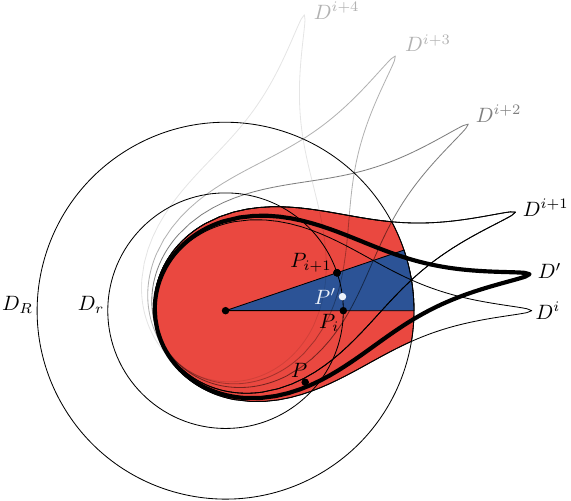}
  \caption{Visualization of the proof of
    Lemma~\ref{lem:arbitrary-angle}.  When constrained to the disk
    $D_R$, the disk $D'$ (bold) with center $P'$ at radius $r$ is
    completely contained in two consecutive disks $D^i$ and $D^{i +
      1}$ (red region).  Point $P_i$ is between $P$ and $P'$.}
  \label{fig:disk-cover}
\end{figure}

\begin{lemma}
  \label{lem:arbitrary-angle}
  Let $G$ be a hyperbolic random graph, let $X_w \ge 0$ for $w \in V$
  be random variables, and let $X(D) = \sum_{w \in D}X_w$ for $D
  \subseteq D_{R}$.  Further, let $\boldsymbol{D}_{R}(r)$ be the set
  of disks of radius $R$ with center at radius $r$.  If for each $D
  \in \boldsymbol{D}_{R}(r)$ it holds that $\Pr[X(D) \le f(n)] \ge 1 -
  p$, then $\Pr[\forall D \in \boldsymbol{D}_{R}(r) \colon X(D) \le
  2f(n)] \ge 1 - \mathcal{O}(np)$.
\end{lemma}
\begin{proof}
  Let $D' \in \boldsymbol{D}_R(r)$ be a disk with radius $R$ centered
  at radius $r$ and arbitrary angular coordinate.  To bound~$X(D')$,
  we cover the disk $D_{R}$ with a circular sequence of $n'$ disks
  $D^1, \dots, D^{n'}$, such that $D'$ is completely contained in two
  consecutive disks (when constrained to the whole disk $D_{R}$).
  That is, there exists an $i \in \{1, \dots, n'\}$ such that $D'
  \subseteq D^i \cup D^{i + 1}$.  Since $X_w \ge 0$ for all $w \in V$,
  it then holds that
  \begin{align*}
    X(D') = \sum_{w \in D'} X_w \le \sum_{D^i \cup D^{i + 1}} X_w \le \sum_{w \in D^i} X_w + \sum_{w \in D^{i+1}} X_w = X(D^i) + X(D^{i + 1}).
  \end{align*}
  Since $\Pr[X(D) \le f(n)] \ge 1 - p$ holds for each $D \in
  \boldsymbol{D}_{R}(r)$, we can apply the union bound to conclude
  that $X(D^i) \le f(n)$ holds for all $i \in \{0, \dots, n'\}$ with
  probability $1 - n'p$.  Consequently, $X(D') \le 2f(n)$ with
  probability $1 - n'p$.

  To complete the proof, it remains to show that there exists such a
  sequence $D^1, \dots, D^{n'}$ with $n' \in \mathcal{O}(n)$.  See
  Figure~\ref{fig:disk-cover} for an illustration of how the sequence
  is constructed.  All disks $D^i$ for $i \in \{1, \dots, n'\}$ have
  their center at radius $r$.  The center of the first disk is placed
  at angular coordinate $0$ and each subsequent disk is placed at an
  angular distance of $2\theta(r, R)$ (see
  Equation~\eqref{eq:maxAdjacentAngle}) to its predecessor until the
  whole disk is covered.  Note that, as a consequence, the boundaries
  of two consecutive disks intersect at the boundary of the whole disk
  $D_{R}$.

  Let $P'$ be the center of $D'$.  To see that $D'$ is contained in
  two consecutive disks $D^i$ and $D^{i + 1}$ (when constrained to the
  whole disk $D_{R}$), first note that there exists an $i \in \{1,
  \dots, n'\}$ such that $P'$ is between the centers $P_i$ and $P_{i +
    1}$ of two consecutive disks $D^i$ and $D^{i + 1}$.  We show that
  any point $P \in D'$ is contained in $D^i \cup D^{i + 1}$.  Clearly,
  $D^i \cup D^{i + 1}$ contains all points between $P_i$ and $P_{i +
    1}$ (blue region in Figure~\ref{fig:disk-cover}).  For the case
  where $P$ does not lie between $P_i$ and $P_{i + 1}$, assume without
  loss of generality, that $P_i$ is between $P$ and $P'$, as depicted
  in Figure~\ref{fig:disk-cover}.  Since $\dist(P, P') \le R$ and
  since $P'$ and $P_i$ have the same radius but $P_i$ is between $P$
  and $P'$, it follows that $\dist(P, P_i) \le R$, and thus $P \in
  D^i$.  Finally, it remains to show that $n' = \mathcal{O}(n)$ disks
  are sufficient to cover the whole disk $D_{R}$.  Since two
  consecutive disks are placed at an angular distance of $2\theta(r,
  R)$, we need $n' = 2\pi / (2\theta(r, R)) = \mathcal{O}(1/\theta(r,
  R))$ disks.  Since $\theta(r, R) \ge \theta(R, R)$, it follows that
  $n' = \mathcal{O}(1/\theta(R, R)) = \mathcal{O}(e^{R/2})$ due to
  Equation~\eqref{eq:maxAdjacentAngle}.  Substituting $R = 2\log(n) +
  C$ then yields the claim.
\end{proof}

\section{Bidirectional Breadth-First Search}
\label{sec:bbbfs}

\begin{figure}[t]
  \centering
  \includegraphics{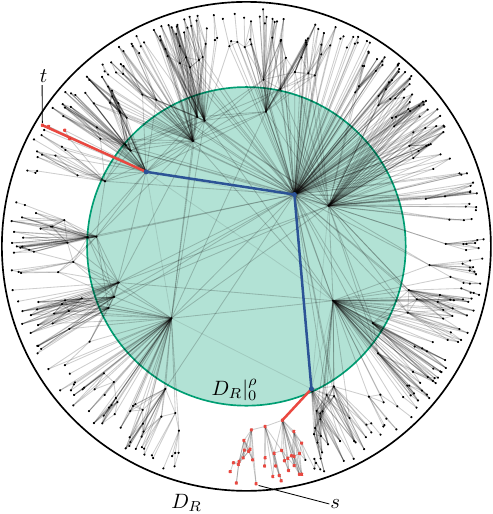}
  \caption{Visualization of the two phases of each BFS in a hyperbolic
    random graph.  Vertices that are visited during the first phase
    are red.  The red edges denote the first encounter of a vertex in
    the inner disk $D_R|_0^\rho$ (green region).  This corresponds to
    the first step in the second phase.  The last step then leads to a
    common neighbor via the blue edges.}
  \label{fig:twoPhases}
\end{figure}

In this section, we analyze the running time of the bidirectional BFS
and obtain an upper bound on the maximum running time over all
possible start--destination pairs.  Our results are summarized in the
following main theorem.

\begin{theorem}
  \label{thm:main-theorem}
  Let $G$ be a hyperbolic random graph.  With high probability the
  shortest path between any two vertices in~$G$ can be computed in
  $\mathcal{\tilde O}(n^{2-1/\alpha} + n^{1/(2\alpha)} +
  \delta_{\max})$ time.
\end{theorem}

We note that this bound on the running time also holds in expectation.
Our bound fails with probability $\mathcal{O}(1/n)$, in which case the
worst-case running time is still bounded by the size of the hyperbolic
random graph, which is $\mathcal{O}(n)$.  Consequently, this case
contributes $\mathcal{O}(1)$ to the expectation, which is dominated by
the above bound.

To prove Theorem~\ref{thm:main-theorem}, we make use of the hyperbolic
geometry in the following way; see Figure~\ref{fig:twoPhases}.  As
long as the two searches visit only low-degree vertices, all explored
vertices lie within a small region, i.e., the searches operate
locally.  Once the searches visit high-degree vertices closer to the
center of the hyperbolic disk (green area in
Figure~\ref{fig:twoPhases}), it takes only few steps to complete the
search, as hyperbolic random graphs have a densely connected core.
Thus, we split our analysis in two phases: a first phase in which both
searches advance towards the center and a second phase in which both
searches meet in the center.  Note that this strategy assumes that we
know the coordinates of the vertices as we would like to stop a search
once it reached the center.  To resolve this issue, we first show in
Section~\ref{sec:bidir-search-altern} that there exists an alternation
strategy that is oblivious to the geometry but performs not much worse
than any other alternation strategy.  We note that this result is
independent of hyperbolic random graphs and thus interesting in its
own right.  Afterwards, in Section~\ref{sec:bidirectional-euclidean}
we examine the performance of the bidirectional BFS on Euclidean
random graphs, before focusing on hyperbolic random graphs in
Section~\ref{sec:bidir-search-hyperb}.

\subsection{Bidirectional Search and Alternation Strategies}
\label{sec:bidir-search-altern}

In an unweighted and undirected graph $G = (V, E)$, a BFS finds the
shortest path between two vertices $s, t \in V$ by starting at $s$ and
exploring the graph in levels, where the $i$th level $L^s_i$ contains
the vertices with distance $i$ to $s$.  More formally, the BFS starts
with the set $L^s_0 = \{s\}$ on level $0$.  Assuming the levels
$L^s_0, \dots, L^s_i$ have been computed already, one obtains the next
level $L^s_{i+1}$ as the set of neighbors of vertices in level $L^s_i$
that are not contained in earlier levels.  Computing $L^s_{i+1}$ from
$L^s_i$ is called an \emph{exploration step}, obtained by
\emph{exploring the edges} between vertices in $L^s_i$
and~$L^s_{i+1}$.

The bidirectional BFS runs two BFSs simultaneously.  The \emph{forward
  search} starts at $s$ and the \emph{backward search} starts at~$t$.
The shortest path between the two vertices can then be obtained, once
the search spaces of the forward and backward search touch.  Since the
two searches cannot actually be run simultaneously, they alternate
depending on their progress.  When exactly the two searches alternate
is determined by the \emph{alternation strategy}.  Note that we only
swap after full exploration steps, i.e., we never explore only half of
level $i$ of one search before continuing with the other.  This has
the advantage that we can be certain to know the shortest path once a
vertex is found by both searches.

In the following we define the \emph{greedy alternation strategy} as
introduced by Borassi and Natale~\cite{bn-kaabra-16} and show that it
is not much worse than any other alternation strategy. Assume the
latest levels of the forward and backward searches are $L^s_i$ and
$L^t_j$, respectively. Then the next exploration step of the forward
search would cost time proportional to $c^s_i \coloneqq \sum_{v\in
  L^s_i} \deg(v)$, while the cost for the backward search is $c^t_j
\coloneqq \sum_{v\in L^t_j} \deg(v)$. The \emph{greedy alternation
  strategy} then greedily continues with the search that causes the
fewer cost in the next exploration step, i.e., it continues with the
forward search if $c^s_i \le c^t_j$ and with the backward search
otherwise.

\begin{theorem}
  \label{thm:alternation-strategy-does-not-matter}
  Let $G$ be a graph with diameter $d$.  If there exists an
  alternation strategy such that the bidirectional BFS explores $f(n)$
  edges, then the bidirectional BFS with greedy alternation strategy
  explores at most $d\cdot f(n)$ edges.
\end{theorem}
\begin{proof}
  Let $A$ be the alternation strategy that explores only $f(n)$ edges.
  First note that the number of explored edges only depends on the
  number of levels explored by the two different searches and not on
  the actual order in which they are explored. Thus, if the greedy
  alternation strategy is different from $A$, we can assume without
  loss of generality that the greedy strategy performed more
  exploration steps in the forward search and fewer in the backward
  search compared to $A$.  Let $c^s$ and $c^t$ be the number of edges
  explored by the forward and backward search, respectively, when
  using the greedy strategy.  Moreover, let $j$ be the last level of
  the backward search (which is actually not explored) and,
  accordingly, let $c^t_j$ be the number of edges the next step in the
  backward search would have explored.  Then $c^t + c^t_j \le f(n)$ as,
  when using $A$, the backward search still explores level $j$.
  Moreover, the forward search with the greedy strategy explores at
  most $c^t + c^t_j$ (and therefore at most~$f(n)$) edges in each
  step, as exploring the backward search would be cheaper otherwise.
  Consequently, each step in the forward and backward search costs at
  most $f(n)$.  As there are at most $d$ steps in total, we obtain the
  claimed bound.
\end{proof}

\subsection{Bidirectional Search in Euclidean Random Graphs}
\label{sec:bidirectional-euclidean}

Euclidean random graphs, commonly known as \emph{random geometric
  graphs}, are generated by distributing $n$ vertices uniformly at
random in the unit square $[0, 1]^2$ and connecting any two vertices
if the Euclidean distance between them is at most some threshold $R
\in \mathbb{R}$~\cite{p-rgg-03}.  One can imagine, that each vertex is
equipped with a disk of radius $R$ and an edge is added to all other
vertices that lie in this disk.  The threshold $R$ affects the
properties of the generated network and in order to obtain graphs with
a giant component of linear size (as is the case for hyperbolic random
graphs), $R$ has to be chosen from the so called \emph{supercritical
  regime}~\cite{p-rgg-03}.  In contrast to hyperbolic random graphs,
the uniform sampling of the vertices in the Euclidean space leads to a
distribution where the number of vertices falling into each disk is
roughly the same, which in turn leads to a homogeneous degree
distribution.

We examine how a BFS explores such a graph, by considering the region
of the plane containing the vertices visited after several exploration
steps.  For Euclidean random graphs with $R$ chosen from the
supercritical regime it was shown that for two vertices at graph
theoretic distance $d$, it holds that $R \cdot d$ is at most a
constant factor larger than the Euclidean distance between them, if
$d$ is super-logarithmic~\cite{fss-dbtrg-13}.  Additionally, it is
easy to see that the Euclidean distance between them can be at most $R
\cdot d$.  Therefore, we can assume that after $k$ (sufficiently many)
steps the region in the plane that contains the visited vertices
resembles a disk of radius proportional to $k$.  Since the area of a
disk with radius $r$ grows as $\pi r^2$, the expected number of
explored vertices is in $\Theta(n k^2)$ (since the vertices are
distributed uniformly).

In this scenario it is easy to see that the performance of a
bidirectional BFS improves by a constant factor, compared to a
standard BFS.  Let $s$ and $t$ be two vertices with (sufficiently
large) graph theoretic distance $d$ from each other.  Then, the
expected number of vertices explored by a standard BFS from $s$ to $t$
is $\Theta(n d^2)$.  If we run two searches instead (one starting at
$s$, the other at $t$), then the expected explored search space is
minimized when the two BFSs touch after half as many steps, exploring
two disks of half the radius.  (Note that this holds independent of
the chosen alternation strategy.)  In that case the expected number of
explored vertices is proportional to $2 n \pi (d/2)^2$ which is again
$\Theta(n d^2$), indicating that the bidirectional variant yields no
asymptotic speedup over the standard BFS.

In the remainder of this paper we focus on the performance of the
bidirectional BFS on hyperbolic random graphs.  In contrast to
Euclidean random graphs, they feature a heterogeneous degree
distribution, leading to significant differences in the performance of
the bidirectional BFS.

\subsection{Bidirectional Search in Hyperbolic Random Graphs}
\label{sec:bidir-search-hyperb}

To analyze the size of the search space of the bidirectional BFS in
hyperbolic random graphs, we separate the whole disk~$D_R$ into two
parts.  One is the \emph{inner disk} $D_R|_0^\rho$ centered at the
origin.  Its radius $\rho$ is chosen in such a way that any two
vertices in $D_R|_0^\rho$ have a common neighbor with high
probability.  The second part is the \emph{outer band} $D_R|_\rho^R$,
the remainder of the whole disk.  A single BFS now explores the graph
in two phases.  In the first phase, the BFS explores vertices in the
outer band.  The phase ends, when the next vertex to be encountered
lies in the inner disk.  Once both BFSs completed the first phase,
they only need at most two more steps for their search spaces to share
a vertex.  One step to encounter the vertex\footnote{Note that this
  vertex has a degree of $\tilde{\Omega}(n^{1 - 1/(2\alpha)})$ with
  high probability.  Consequently, a non-giant component of size
  $\tilde{\mathcal{O}}(1)$ is detected (at the latest) before
  exploring this vertex (see Section~\ref{sec:preliminaries}). } in
the inner disk and another step to meet at their common neighbor that
any two vertices in the inner disk have with high probability; see
Figure~\ref{fig:twoPhases}.

Note that this scenario describes the worst case.  Depending on the
positions of the two considered vertices the two searches may touch
earlier, e.g., when both vertices are close to each other in the outer
band or when at least one of them is already contained in the inner
disk.  However, since we want to determine an upper bound on the
running time, we consider the case where both vertices lie in the
outer band and the two searches touch in the inner disk.  In the
remainder of the paper we only consider how one of the two searches
explores the graph.  The obtained bounds also hold for the other
search, meaning the total search space increases only by a constant
factor when considering both searches instead of only one.

For our analysis we assume an alternation strategy in which each
search stops once it explored one additional level after finding the
first vertex in the inner disk $D_R|_0^\rho$.  Of course, this cannot
be implemented without knowing the underlying geometry of the network.
However, by Theorem~\ref{thm:alternation-strategy-does-not-matter} the
search space explored using the greedy alternation strategy is only a
poly-logarithmic factor larger, as the diameter of hyperbolic random
graphs is poly-logarithmic with high
probability~\cite{fk-dhrg-18}\footnote{We note that there is a tighter
  bound of $\mathcal{O}(\log(n))$ on the diameter of hyperbolic random
  graphs, which holds asymptotically almost surely~\cite{ms-k-19}.}.
The following lemma shows for which choice of~$\rho$ the above
sketched strategy works.

\begin{lemma}
  \label{lem:innerDiskRadius}
  Let $G$ be a hyperbolic random graph.  With high probability, $G$
  contains a vertex that is adjacent to every other vertex in
  $D_R|_0^\rho$, for $\rho = \frac{1}{\alpha}(\log n - \log\log n)$.
\end{lemma}
\begin{proof}
  Assume $v$ is a vertex with radius at most $R - \rho$.  Note that
  the distance between two points is upper bounded by the sum of their
  radii.  Thus, every vertex in $D_R|_0^\rho$ has distance at most $R$
  to $v$, and is therefore adjacent to $v$.  Hence, to prove the
  claim, it suffices to show the existence of this vertex $v$ with
  radius at most $R - \rho$.  As described in
  Section~\ref{sec:preliminaries}, the probability for a single vertex
  to have radius at most $R - \rho$ is given by the
  measure~$\mu(D_R|_0^{R - \rho})$.  Using
  Equation~\eqref{eq:originBallApproximation} we obtain
  \begin{align*}
    \mu(D_R|_0^{R - \rho}) &= e^{-\alpha \rho}(1 + o(1)) = \frac{\log n}{n}(1 + o(1)).
  \end{align*}
  Thus, the probability that none of the $n$ vertices lies in
  $D_R|_0^{R-\rho}$ is given by $(1 - \mu(D_R|_0^{R - \rho}))^n$.
  That is,
  \begin{align*}
    \Pr[\{v \in D_R|_0^{R - \rho}\} = \emptyset] = \left(1 - \frac{\log n}{n} (1 + o(1)) \right)^n.
  \end{align*}
  Since $(1 - x) \le e^{-x}$ for all $x \in \mathbb{R}$, this term can
  be bounded by
  \begin{align*}
    \Pr[\{v \in D_R|_0^{R - \rho}\} = \emptyset ] &\le e^{-\frac{\log(n)}{n}(1 + o(1)) \cdot n} = e^{-\log(n)(1 + o(1))} = n^{-(1 + o(1))} = \mathcal{O}(1/n).
  \end{align*}
  Hence, there is at least one vertex with radius at most $R - \rho$
  with high probability.
\end{proof}

In the following, we first bound the search space explored in the
first phase, i.e., before we enter the inner disk $D_R|_0^\rho$.
Afterwards we bound the search space explored in the second phase,
which consists of two exploration steps.  The first one to enter
$D_R|_0^\rho$ and the second one to find a common neighbor, which
exists due to Lemma~\ref{lem:innerDiskRadius}.

\subsubsection{Search Space in the First Phase} 
\label{sec:search-space-first}

To bound the size of the search space in the outer band, we make use
of the geometry in the following way.  For two vertices in the outer
band to be adjacent, their angular distance has to be small.
Moreover, the number of exploration steps is bounded by the diameter
of the graph.  Thus, the maximum angular distance between vertices
visited in the first phase cannot be too large.  Note that the
following lemma restricts the search to a sublinear portion of the
disk, which we later use to show that also the number of explored
edges is sublinear.

\begin{figure}
  \centering
  \includegraphics{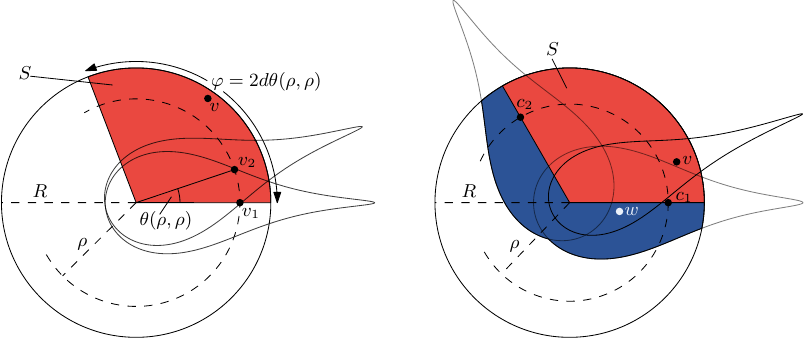}
  \caption{Left: The sector $S$ (red) of angular width $\varphi$
    contains the search space of a BFS starting at $v$, in the outer
    band~$D_R|_\rho^R$.  The vertices $v_1$ and $v_2$ are at maximum
    angular distance to still be adjacent.  Right: Neighbor $w$ of
    vertex $v$ is in $S$ (red) or a neighbor of $c_1$ or $c_2$
    (blue).}
  \label{fig:outerBand}
\end{figure}

\begin{lemma}
  \label{lem:angularInterval}
  With high probability, all vertices that a BFS on a hyperbolic
  random graph explores before finding a vertex with radius at most
  $\rho = \frac{1}{\alpha}(\log n - \log\log n)$ lie within a sector
  of angular width $\mathcal{\tilde O}(n^{-(1/\alpha-1)})$.
\end{lemma}
\begin{proof}
  For an illustration of the proof see
  Figure~\ref{fig:outerBand}~(left).  Recall from
  Section~\ref{sec:preliminaries} that $\theta(r_1, r_2)$ denotes the
  maximum angular distance between two vertices of radii $r_1$ and
  $r_2$ such that they are still adjacent.  Since $r_1$ and $r_2$ only
  appear as negative exponents in the expression for $\theta(r_1,
  r_2)$ (see Equation~\eqref{eq:maxAdjacentAngle}), this angle
  increases with decreasing radii.  Thus, $\theta(r_1, r_2) \le
  \theta(\rho, \rho)$ holds for all vertices in the outer band
  $D_R|_\rho^R$.

  Now assume we start a BFS at a vertex $v \in D_R|_\rho^R$ and
  perform $d$ exploration steps without leaving the outer
  band~$D_R|_\rho^R$.  Then no explored vertex has angular distance
  more than $d\theta(\rho, \rho)$ from $v$.  Thus, the whole search
  space lies within a disk sector of angular width $2d\theta(\rho,
  \rho)$.  The number of steps $d$ is at most poly-logarithmic as the
  diameter of a hyperbolic random graph is poly-logarithmic with high
  probability~\cite{fk-dhrg-18}.  Using
  Equation~\eqref{eq:maxAdjacentAngle} for $\theta(\rho,\rho)$, we
  obtain
  \begin{align*}
    \theta(\rho, \rho) &= 2e^{\frac{R - 2\rho}{2}}(1 + \Theta(e^{R - 2\rho})) \\
                       &= 2e^{C/2}n^{1-1/\alpha}\log(n)^{1/\alpha}(1 + \Theta((\log n/n^{1-\alpha})^{2/\alpha})) \\
                       &= \mathcal{O}(n^{-(1/\alpha-1)}\log(n)^{1/\alpha}),
  \end{align*}
  which proves the claimed bound.
\end{proof}

Note that the expected number of vertices in a sector $S$ of angular
width $\varphi$ is linear in~$n\varphi$ due to the fact that the
angular coordinate of each vertex is chosen uniformly at random.
Thus, Lemma~\ref{lem:angularInterval} already shows that the expected
number of vertices visited in the first phase of the BFS is $\mathcal
{\tilde O}(n^{2-1/\alpha})$, which is sublinear in $n$.  It is also
not hard to see that this bound holds with high probability (see
Corollary~\ref{col:chernoff-bound}).  To also bound the number of
explored edges, we sum the degrees of vertices in $S$.  It is not
surprising that this yields the same asymptotic bound in expectation,
as the expected average degree in a hyperbolic random graph is
constant.  However, to obtain meaningful results, we need a bound that
holds with high probability.  Though we can use techniques similar to
those that have been used to show that the average degree of the whole
graph is constant with high
probability~\cite{bkl-adgcs-16,gpp-rhg-12}, the situation is
complicated by the restriction to a sublinear portion of the disk.
Nonetheless, we obtain the following theorem.

\wormhole{main-theorem-first-phase}
\newcommand{\mainTheoremFirstPhaseText}{Let $G$ be a hyperbolic random
  graph.  The degrees of vertices in every sector of angular width
  $\varphi$ sum to $\mathcal{\tilde O}(\varphi n + n^{1/(2\alpha)} +
  \delta_{\max})$ with high probability if $\varphi =
  \Omega(\log(n)^{2} / n^{1/2}).$}
\begin{theorem}
  \label{thm:main-theorem-first-phase}
  \mainTheoremFirstPhaseText
\end{theorem}

We note that $\delta_{\max}$ has to be included here, as the theorem
states a bound for every sector, and thus in particular for sectors
containing the vertex of maximum degree.  Recall, that $\delta_{\max}
= \mathcal {\tilde O}(n^{1/(2\alpha)})$ holds almost
surely~\cite{gpp-rhg-12}.  Moreover, we note that the condition
$\varphi = \Omega(\log(n)^{2} / n^{1/2})$ is crucial for our proof,
i.e., the angular width of the sector has to be sufficiently large for
the concentration bound to hold.  We note that, depending on $\alpha$,
the angular width determined in Lemma~\ref{lem:angularInterval} may be
smaller than this lower bound.  However, if this is the case, we can
choose $\varphi = \mathcal{\tilde{O}}(n^{-1/2})$ as an upper bound for
the angular width of the sector and obtain
$\mathcal{\tilde{O}}(\varphi n) = \mathcal{\tilde{O}}(n^{1/2}) =
\mathcal{\tilde{O}}(n^{1/(2\alpha)})$ for $\alpha \in (1/2, 1)$.
Consequently, the bound holds for the previously determined angular
width $\mathcal{\tilde{O}}(n^{-(1/\alpha - 1)})$ for all $\alpha \in
(1/2, 1)$.

As the proof for Theorem~\ref{thm:main-theorem-first-phase} is rather
technical, we defer it to Section~\ref{sec:conc-bounds-sum-deg}.
Together with Lemma~\ref{lem:angularInterval}, we obtain the following
corollary.  Note that since $\alpha \in (1/2, 1)$, this shows that the
running time spend in the first phase (not accounting for the maximum
degree) is sublinear in $n$ with high probability.

\begin{corollary}
  \label{cor:firstPhase}
  On a hyperbolic random graph, the first phase of the bidirectional
  search explores with high probability only $\mathcal {\tilde
    O}(n^{2-1/\alpha} + n^{1/(2\alpha)} + \delta_{\max})$ many edges.
\end{corollary}

\subsubsection{Search Space in the Second Phase}
\label{search-space-in-second-phase}

The first phase of the BFS is completed when the next vertex to be
encountered lies in the inner disk. Thus, the second phase consists of
only two exploration steps.  One step to encounter the vertex in the
inner disk and another step to meet the other search.  Thus, to bound
the running time of the second phase, we have to bound the number of
edges explored in these two exploration steps.  To do this, let $V_1$
be the set of vertices encountered in the first phase.  Recall that
all these vertices lie within a sector $S$ of angular width $\varphi =
\mathcal {\tilde O}(n^{-(1/\alpha - 1)})$
(Lemma~\ref{lem:angularInterval}).  The number of explored edges in
the second phase is then bounded by the sum of degrees of all
neighbors $N(V_1)$ of vertices in $V_1$.  To bound this sum, we divide
the neighbors of $V_1$ into two categories: $N(V_1) \cap S$ and
$N(V_1) \setminus S$.  Note that we already bounded the sum of degrees
of vertices in $S$ for the first phase (see
Theorem~\ref{thm:main-theorem-first-phase}), which clearly also bounds
this sum for $N(V_1) \cap S$.  Thus, it remains to bound the sum of
degrees of vertices in $N(V_1) \setminus S$.

To bound this sum, we introduce two \emph{hypothetical vertices}
(i.e., vertices with specific positions that are not actually part of
the graph) $c_1$ and $c_2$ such that every vertex in
$N(V_1) \setminus S$ is a neighbor of $c_1$ or $c_2$.  Then it remains
to bound the sum of degrees of neighbors of these two vertices.  To
define $c_1$ and $c_2$, recall that the first phase was restricted to
points in the sector $S$ that have a radius greater than $\rho$, i.e.,
all vertices in $V_1$ lie within $S|_\rho^R$.  The hypothetical
vertices $c_1$ and $c_2$ are basically positioned at the corners of
this region, i.e., they both have radius $\rho$, and they assume the
maximum and minimum angular coordinate within~$S$, respectively.
Figure~\ref{fig:outerBand} (right) shows these positions.  We obtain
the following.
\begin{lemma}
  \label{lem:boundaryCover}
  Let $G$ be a hyperbolic random graph, let $S$ be a sector, and let
  $v \in S|_\rho^R$ be a vertex.  Then, every neighbor of $v$ lies in
  $S$ or is a neighbor of one of the hypothetical vertices~$c_1$ or
  $c_2$.
\end{lemma}
\begin{proof}
  Let $v = (r, \varphi) \in S|_\rho^R$ and $w \in N(v)\setminus S$.
  Without loss of generality, assume that $c_1$ lies between $v$ and
  $w$, as is depicted in Figure~\ref{fig:outerBand} (right). Now
  consider the point $v' = (\rho, \varphi)$ obtained by moving $v$ to
  the same radius as $c_1$.  According to
  Lemma~\ref{lem:smallerRadiusIncreasesNeighborhood} we have $N(v)
  \subseteq N(v')$.  In particular, it holds that $w \in N(v')$ and
  therefore $\dist(v', w) \le R$.  Since $v'$ and $c_1$ have the same
  radial coordinate and $c_1$ is between $v'$ and $w$, it follows that
  $\dist(c_1, w) \le R$.
\end{proof}

By the above argument, it remains to sum the degrees of neighbors of
$c_1$ and $c_2$.  In the following, we show that the degrees of the
neighbors of a vertex with radius $r$ sum
to~$\Theta(ne^{-(\alpha - 1/2)r})$ in expectation.  We note that, for
large values of $r$, i.e., for a vertex lying close to the boundary of
the disk, this term is surprisingly large.  This is due to the fact
that, although vertices near the center of the disk are rather
unlikely to exist in the first place, their degree would be
sufficiently large such that they dominate the expected degree sum.

\begin{lemma}
  \label{lem:expected-neighborhood-degree-sum}
  Let $G$ be a hyperbolic random graph.  The degrees of the neighbors
  of a vertex~$v$ sum to $\Theta(ne^{-(\alpha - 1/2)r(v)})$ in
  expectation.
\end{lemma}
\begin{proof}
  Let $Z_v$ be the sum of the degrees of the neighbors of $v$, which
  is a random variable that depends on the positions of all vertices
  in the graph.  Formally, we can express $Z_v$ by assigning each
  vertex $w \in V \setminus \{v\}$ two random variables $X_w$ and
  $Y_w$.  The first is an indicator random variable with $X_w = 1$ if
  $w$ is a neighbor of $v$ and $X_w = 0$ otherwise.  Additionally, the
  random variable $Y_w$ denotes the degree of $w$.  The sum of the
  degrees of the neighbors of~$v$ can then be written as
  \begin{align*}
    Z_v = \sum_{w \in V \setminus \{v\}} X_w \cdot Y_w.
  \end{align*}
  The expected value of $Z_v$ is given by
  \begin{align*}
    \mathbb{E}[Z_v] = \mathbb{E}\left[ \sum_{w \in V \setminus \{v\}} X_w \cdot Y_w \right] = \sum_{w \in V \setminus \{v\}} \mathbb{E}\left[ X_w \cdot Y_w \right],
  \end{align*}
  where the second equality holds due to the linearity of expectation.
  To compute the expected value of $X_w \cdot Y_w$ we can apply the
  law of total expectation and obtain
  \begin{align*}
    \mathbb{E}[Z_v] = \sum_{w \in V \setminus \{v\}} \sum_{x \in \{0, 1\}} \mathbb{E}[X_w \cdot Y_w \mid X_w = x] \cdot \Pr[X_w = x].
  \end{align*}
  Clearly, the case where $X_w = 0$ does not contribute anything to
  the sum, which can thus be simplified as
  \begin{align*}
    \mathbb{E}[Z_v] = \sum_{w \in V \setminus \{v\}} \mathbb{E}[Y_w \mid X_w = 1] \cdot \Pr[X_w = 1].
  \end{align*}
  Recall that $X_w = 1$ denotes the event where $w$ is a neighbor of
  $v$.  That is,
  \begin{align*}
    \Pr[X_w = 1] = \Pr[w \in N(v)] = \Pr[w \in D_{R}(v)] = \mu(D_{R}(v)).
  \end{align*}
  Moreover, recall that $Y_w$ denotes the random variable representing
  the degree of $w$.  Consequently, we can now write $\mathbb{E}[Z_v]$
  as
  \begin{align}
    \label{eq:expected-degree-sum-intermediate}
    \mathbb{E}[Z_v] &= \sum_{w \in V \setminus \{ v \}} \mu(D_R(v)) \cdot \mathbb{E}[Y_w \mid w \in D_R(v)] \notag \\
                    &= (n - 1) \cdot \mu(D_R(v)) \cdot \mathbb{E}[\deg(w) \mid w \in D_R(v)].
  \end{align}
  We continue by computing the expected degree of a vertex $w$
  \emph{conditioned} on the fact that it is contained in $D_R(v)$.  To
  this end, we first consider the expected value without the
  condition, analogous to how it was done previously~\cite{gpp-rhg-12}
  (see the proof of Theorem 2.3), and afterwards explain how to
  incorporate the condition.  The expected degree of a vertex $w$ with
  fixed radius $r$ is given by
  \begin{align*}
    \mathbb{E}[\deg(w) \mid r(w) = r] = (n - 1) \mu(D_{R}(w)).
  \end{align*}
  To obtain the expected degree of $w$ without
  fixing its radius (or angle for that matter) we then integrate
  $\mathbb{E}[\deg(w) \mid r(w) = r \wedge \varphi(w) = \varphi] \cdot
  f(r, \varphi)$ (note the joint distribution) over the whole
  disk. That is,
  \begin{align*}
    \mathbb{E}[\deg(w)] &= \iint_{D_R} \mathbb{E}[\deg(w) \mid r(w) = r \wedge \varphi(w) = \varphi] \cdot f(r, \varphi) \dif \varphi \dif r \\
                        &= \iint_{D_R} \mathbb{E}[\deg(w) \mid r(w) = r] \cdot f(r, \varphi) \dif \varphi \dif r,
  \end{align*}
  where the second step follows from the fact that the expected degree
  of a vertex is independent of its angular coordinate.

  It remains to include the condition on the fact that $w$ cannot be
  anywhere in the whole disk but lies in $D_R(v)$ instead.  First, we
  have to accommodate for the fact that if $w$ is a neighbor of $v$,
  then conversely $v$ is also a neighbor of~$w$.  Consequently, we
  know that $w$ has at least one neighbor, which we reflect in the
  expected value by introducing the condition on the position $p(v)$
  of $v$.  Moreover, in general the conditional expectation of a
  random variable~$X$ conditioned on an event $A$ (with $\Pr[A] > 0$)
  is given by
  $\mathbb{E}[X \mid A] = \int_{-\infty}^{\infty} x f_{X \mid A}(x)
  \dif x$, where~$f_{X \mid A}$ is defined as
  \begin{align*}
    f_{X \mid A}(x) =
    \begin{cases}
      \frac{f_X(x)}{\Pr[A]}, &x \in A,\\
      0, &x \notin A.
    \end{cases}
  \end{align*}
  Therefore, the above expression for the expected degree of $w$ can
  be adjusted to include the condition as
  \begin{align*}
    \mathbb{E}[\deg(w) \mid w \in D_R(v)] &= \iint_{D_R(v)} \mathbb{E}[\deg(w) \mid r(w) = r \wedge p(v) = (r, \varphi)] \\
                                          &\hphantom{= \iint_{D_R(v)}~} \cdot \frac{f(r, \varphi)}{\Pr[w \in D_R(v)]} \dif \varphi \dif r.
  \end{align*}
  Note that the probability $\Pr[w \in D_{R}(v)]$ in the denominator
  is, again, the measure of the disk of radius $R$ centered at $v$.
  Substituting this expression in
  Equation~\eqref{eq:expected-degree-sum-intermediate} for the
  expected sum $\mathbb{E}[Z_v]$ of the degrees of the neighbors of
  $v$, we get
  \begin{align*}
    \mathbb{E}[Z_v] &= (n - 1) \cdot \mu(D_R(v)) \cdot \mathbb{E}[\deg(w) \mid w \in D_R(v)] \\
                    &= (n - 1) \cdot \mu(D_R(v)) \cdot \\
                    &\hphantom{= (n - 1 \cdot~)} \cdot \iint_{D_R(v)} \mathbb{E}[\deg(w) \mid r(w) = r \wedge p(v) = (r, \varphi)] \frac{f(r, \varphi)}{\mu(D_R(v))} \dif \varphi \dif r \\
                    &= (n - 1) \cdot \iint_{D_R(v)} \mathbb{E}[\deg(w) \mid r(w) = r \wedge p(v) = (r, \varphi)] \cdot f(r, \varphi) \dif \varphi \dif r.
  \end{align*}
  To compute the integral, we determine the expected degree of $w$
  conditioned on the fact that $r(w) = r$ and on the position of $v$,
  which is a neighbor of $w$ deterministically.  Therefore, we obtain
  the expected degree by adding $1$ (for $v$) to the expected number
  of vertices among the remaining $V \setminus \{v, w\}$ that are
  sampled into $D_{R}(w)$ and obtain
  \begin{align*}
    \mathbb{E}[\deg(v) \mid r(w) = r \wedge p(v) = (r, \varphi)] = 1 + (n - 2)\mu(D_{R}(w)),
  \end{align*}
  which is $1 + \Theta(ne^{-r/2})$ due to
  Equation~\eqref{eq:expectedDegree}.  Note that $\Theta(ne^{-r/2})$
  is $\Omega(1)$ for all $r \in [0, R]$, allowing us to further
  simplify the expected value to
  \begin{align*}
    \mathbb{E}[\deg(v) \mid r(w) = r \wedge p(v) = (r, \varphi)] = \Theta(ne^{-r/2}).
  \end{align*}
  Moreover, recall that $f(r, \varphi) = 0$ for $r > R$ and that it
  can otherwise be bounded by $f(r, \varphi) = \Theta(e^{-\alpha(R -
    r)})$ (see Equation~\eqref{eq:probDensity}).  We obtain
  \begin{figure}
    \centering
    \includegraphics{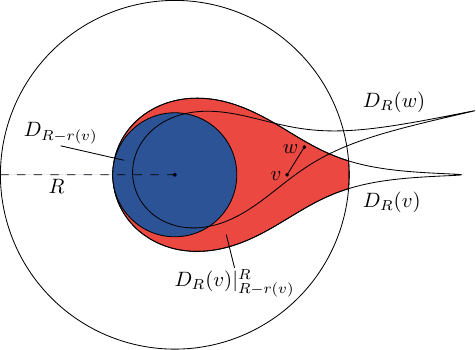}
    \caption{Situation in the proof of
      Lemma~\ref{lem:expected-neighborhood-degree-sum}.  Vertex $w$ is
      a neighbor of $v$.  To integrate $D_R(v) \cap D_R$, we split the
      region into two parts: $D_R(v)|_0^{R - r(v)} = D_{R - r(v)}$
      (blue) and $D_R(v)|_{R - r(v)}^R$ (red).}
    \label{fig:neighbor-degree}
  \end{figure}
  \begin{align*}
    \mathbb{E}[Z_v] &= \Theta \left( (n - 1) \cdot \iint_{D_R(v) \cap D_R} ne^{-r/2} \cdot e^{-\alpha(R - r)} \dif \varphi \dif r \right) \\
                    &= \Theta \left( n^2 e^{-\alpha R} \cdot \iint_{D_R(v) \cap D_R} e^{(\alpha - 1/2)r} \dif \varphi \dif r \right).
  \end{align*}
  We can now split the integral into two parts: one containing the
  disk $D_R(v)|_0^{R - r(v)} = D_{R - r(v)}$ and the other containing
  the remainder of $D_R(v) \cap D_R$, which is given by $D_R(v)|_{R -
    r(v)}^{R}$ (see Figure~\ref{fig:neighbor-degree}).  For the second
  part we can use Equation~\eqref{eq:maxAdjacentAngle} to bound the
  angle $\theta(r(v), r)$ up to which we need to integrate depending
  on $r$.  As a result, we get
  \begin{align*}
    \mathbb{E}[Z_v] &= \Theta \Bigg( n^2 e^{-\alpha R} \cdot \Bigg( \int_{0}^{R - r(v)} \int_{0}^{2\pi} e^{(\alpha - 1/2)r} \dif \varphi \dif r \\
                    &\hphantom{= \Theta \Bigg( n^2 e^{-\alpha R} \cdot \Bigg(~} + \int_{R - r(v)}^{R} \int_{0}^{\theta(r(v), r)} e^{(\alpha - 1/2)r} \dif \varphi \dif r \Bigg) \Bigg).
  \end{align*}
  Regarding the first part of the sum, note that evaluating the inner
  integral only contributes a constant factor that can be dropped due
  the $\Theta$-notation.  Computing the outer integral then yields
  $\Theta(e^{(\alpha - 1/2)(R - r(v))})$.  For the second part of the
  sum we, again, first evaluate the inner integral and substitute
  $\theta(r(v), r) = \Theta(e^{(R - r(v) - r)/2})$ (see
  Equation~\eqref{eq:maxAdjacentAngle}).  We obtain
  \begin{align*}
    \mathbb{E}[Z_v] &= \Theta \left( n^2 e^{-\alpha R} \cdot \left( e^{(\alpha - 1/2)(R - r(v))} + e^{(R - r(v))/2}\int_{R - r(v)}^{R} e^{-(1 - \alpha)r} \dif r \right) \right).
  \end{align*}
  The last integral evaluates to $\mathcal{O}(e^{-(1-\alpha)(R -
    r(v))})$, which multiplied by the factor $e^{(R - r(v))/2}$ yields
  asymptotically the same expression as the first summand and we get
  \begin{align*}
    \mathbb{E}[Z_v] &= \Theta \left( n^2 e^{-(\alpha - 1/2)r(v)} \cdot e^{-R/2} \right).
  \end{align*}
  Finally, we can substitute $R = 2\log(n) + C$ in order to obtain the
  claimed bound of
  $\mathbb{E}[Z_v] = \Theta(ne^{-(\alpha - 1/2)r(v)})$.
\end{proof}

For $c_1$ and $c_2$, which both have radius $\rho$, the degrees of
their neighbors thus sum to~$\mathcal{\tilde O}(n^{1/(2\alpha)})$ in
expectation.  However, to actually prove
Theorem~\ref{thm:main-theorem}, we need a bound that holds with high
probability for all possible angular coordinates of $c_1$ and $c_2$.
As with the sum of the degrees in a sector, we prove a slightly weaker
bound that matches the one in Theorem~\ref{thm:main-theorem} and holds
with high probability.  We obtain the following lemma.

\wormhole{main-theorem-second-phase}
\newcommand{\mainTheoremSecondPhaseText}{Let $G$ be a hyperbolic random graph and let $v$ be a hypothetical
  vertex with radius $\rho = 1/\alpha (\log n - \log\log n)$ and
  arbitrary angular coordinate.  The degrees of neighbors of $v$ sum
  to
  $\mathcal {\tilde O}(n^{2-1/\alpha} + n^{1/(2\alpha)} +
  \delta_{\max})$ with high probability.}
\begin{lemma}
  \label{lem:second-phase-concentration}
  \mainTheoremSecondPhaseText
\end{lemma}

Again, the proof is rather technical and thus deferred to
Section~\ref{sec:conc-bounds-sum-deg}. Together with the bounds on the
sum of degrees in a sector of width $\varphi = \mathcal{\tilde
  O}(n^{-(1/\alpha - 1)})$
(Theorem~\ref{thm:main-theorem-first-phase}), we obtain the following
corollary, which concludes the proof of
Theorem~\ref{thm:main-theorem}.

\begin{corollary}
  \label{cor:twoSteps}
  On a hyperbolic random graph, the second phase of the bidirectional
  search explores with high probability only $\mathcal{\tilde
    O}(n^{2-1/\alpha} + n^{1/(2\alpha)} + \delta_{\max})$ many edges.
\end{corollary}

\section{Concentration Bounds for the Sum of Vertex Degrees}
\label{sec:conc-bounds-sum-deg}

Here we prove the concentration bounds that were announced in the
previous section.  For the first phase, we already know that the
search space is contained within a sector $S$ of sublinear width
(Lemma~\ref{lem:angularInterval}).  Thus, the running time in the
first phase is bounded by the sum of vertex degrees in this sector.
Moreover, all edges explored in the second phase also lie within the
same sector $S$ or are incident to neighbors of the two hypothetical
vertices $c_1$ and $c_2$ (Lemma~\ref{lem:boundaryCover}).  Thus, the
running time of the second phase is bounded by the sum of vertex
degrees in $S$ and in the neighborhood of $c_1$ and $c_2$.

In both cases, we have to bound the sum of vertex degrees in certain
areas of the disk, which can be done as follows.  For each degree, we
want to compute the number of vertices of this degree in the
considered area and multiply it with the degree.  As all vertices with
a certain degree have roughly the same radius, we can separate the
disk into small bands, one for each degree.  Then summing over all
degrees comes down to summing over all bands and multiplying the
number of vertices in this band with the corresponding degree.  If we
can prove that each of these values is highly concentrated (i.e.,
holds with probability $1-\mathcal{O}(n^{-2})$), we obtain that the
sum is concentrated as well (using the union bound).  Unfortunately,
this fails in two situations.  For small radii, the number of vertices
within the corresponding band (i.e., the number of high degree
vertices) is too small to be concentrated.  Moreover, for large radii
the degree is too small to be concentrated around its expected value.

To overcome this issue, we partition the disk $D_{R}$ into three
parts.  An inner part $D_R|_0^{\rho_I(\varphi)}$, containing all
points of radius at most $\rho_I(\varphi)$, an outer part
$D_R|_{\rho_O}^R$, containing all points of radius at least $\rho_O$,
and a central part $D_R|_{\rho_I(\varphi)}^{\rho_O}$, containing all
points in between.  We choose $\rho_I(\varphi)$ such that the number
of vertices with maximum degree in a sector part
$S|_{\rho_I(\varphi)}^{\rho_O}$ of angular width $\varphi$ is
$\Omega(\log n)$, which ensures that for each vertex degree, the
number of vertices with this degree is concentrated.  Moreover, we
choose $\rho_O$ in such a way that the vertex degrees in
$S|_{\rho_I(\varphi)}^{\rho_O}$ are sufficiently concentrated.  To
achieve this, we set
\[\rho_I(\varphi) = R - \frac{1}{\alpha} \left( \log(\varphi / (2\pi)) + \log n - \log\log
    n \right) \text{ and } \, \rho_O = R - (2 +
  \varepsilon)\log\log(n),\] for any constant $\varepsilon \in (0, 1)$, and
show concentration for the sum of the degrees in a sector and in the
neighborhood of a vertex with radius $\rho$, separately for the three
parts of the disk.

\subsection{The Inner Part of the Disk}

The inner part $D_R|_0^{\rho_I(\varphi)}$ contains vertices of high
degree.  It is not hard to see that there are only logarithmically
many vertices with radius at most $\rho_I(\varphi)$.
\begin{lemma}
  \label{lem:sector-polylog-radius}
  Let $G$ be a hyperbolic random graph, let $\varphi \in [0, 2\pi]$ be
  an angle, and let $\xi > 0$ be a constant.  A
  sector~$S|_0^{\rho_I(\varphi)}$ of angular width $\xi\varphi \in [0,
  2\pi]$ contains $\mathcal{O}(\log(n))$ vertices, with probability $1
  - \mathcal{O}(n^{-c})$ for any constant $c$.
\end{lemma}
\begin{proof}
  By Equation~\eqref{eq:originBallApproximation} the expected number
  of vertices in the disk $D_R|_0^{\rho_I(\varphi)}$ is given by
  \begin{align*}
    \mathbb{E}[|\{v \in D_R|_0^{\rho_I(\varphi)}\}|] = ne^{-\alpha(R - \rho_I(\varphi))}(1 + o(1)).
  \end{align*}
  Since the angular coordinates of
  the vertices are distributed uniformly in $[0, 2\pi]$, the expected
  number of vertices in a sector portion~$S|_0^{\rho_I(\varphi)}$ of
  angular width $\xi\varphi$ is
  \begin{align*}
    \mathbb{E}[|\{v \in S|_0^{\rho_I(\varphi)}\}|] &= \frac{\xi\varphi}{2\pi} n e^{-\alpha(R - \rho_I(\varphi))} (1 + o(1)) \\
                                                   &= \frac{\xi\varphi}{2\pi} n e^{-(\log(\varphi / 2\pi) + \log n - \log \log n)} (1 + o(1)) \\
                                                   &= \xi\log (n) (1 + o(1)).
  \end{align*}
  Since $\xi > 0$ is constant, this bound is in $\Omega(\log n)$ and
  we can apply Corollary~\ref{col:chernoff-bound} to conclude that
  $|\{v \in S|_0^{\rho_I(\varphi)}\}| = \mathcal{O}(\log (n))$ holds
  with probability $1 - \mathcal{O}(n^{-c})$ for any constant $c$.
\end{proof}

Note that, if $\varphi \in \Omega(1/n)$, we can choose at most
$\mathcal{O}(n)$ sectors of width $2\varphi$ such that any sector of
width $\varphi$ lies completely in one of them.  Thus, the probability
that there exists a sector portion $S|_0^{\rho_I(\varphi)}$ where the
number of vertices is super-logarithmic, is bounded by the probability
that it is too large in at least one of these $\mathcal{O}(n)$ sectors
(of twice the width).  By choosing $\xi = 2$, we can apply
Lemma~\ref{lem:sector-polylog-radius} to conclude that a single sector
$S|_0^{\rho_I(\varphi)}$ of twice the angular width contains at most
$\mathcal{O}(\log(n))$ vertices with probability $1 -
\mathcal{O}(n^{-2})$.  Applying the union bound and incorporating the
fact that the maximum degree in the graph is $\delta_{\max}$ we can
bound the number of edges in every such sector portion and obtain the
following corollary.

\begin{corollary}
  \label{cor:number_edges_in_inner_sector}
  Let $G$ be a hyperbolic random graph.  For every sector $S$ of
  angular width $\varphi \in \Omega(1/n)$, the degrees of the vertices
  in $S|_0^{\rho_I(\varphi)}$ sum to $\mathcal {\tilde
    O}(\delta_{\max})$ with high probability.
\end{corollary}

Note that, in particular the statement holds for the previously
determined angle $\varphi = \mathcal{\tilde{O}}(n^{-(1/\alpha - 1)})$
for $\alpha \in (1/2, 1)$.  Additionally, by setting $\varphi = 2\pi$,
we can use Lemma~\ref{lem:sector-polylog-radius} to bound the sum of
the degrees of the high degree vertices in the neighborhood of a
vertex with radius $\rho$.

\begin{corollary}
  \label{cor:number_edges_in_inner_neighborhood}
  Let $G$ be a hyperbolic random graph.  For every vertex $v$ of
  radius $\rho$, the degrees of the neighbors of $v$
  in~$D_R|_0^{\rho_I(2\pi)}$ sum to $\mathcal {\tilde
    O}(\delta_{\max})$ with high probability.
\end{corollary}

\subsection{The Central Part of the Disk}

For each possible vertex degree $k$, we want to compute the number of
vertices with this degree in the central
part~$D_R|_{\rho_I(\varphi)}^{\rho_O}$.  First note, that by
Equation~\eqref{eq:expectedDegree} a vertex with fixed radius has
expected degree $\Theta(k)$ if this radius is $2\log(n/k)$.  Motivated
by this, we define $r_k = 2\log(n/k)$.  To bound the sum of degrees in
the central part $D_R|_{\rho_I(\varphi)}^{\rho_O}$, we use that
vertices with radius significantly larger than $r_k$ also have a
smaller degree.  To this end, we first prove that a vertex with degree
$k$ can actually not have a radius much larger than~$r_k$.  This has
the advantage, that we can bound the number of degree-$k$ vertices by
bounding the number of vertices with these radii.

\begin{lemma}
  \label{lem:probDegreeTooLarge}
  Let $G$ be a hyperbolic random graph.  Then, for every constant $c >
  0$, there exist constants $\kappa, \tau > 0$, such that all vertices
  with degree at least $k \ge \kappa \log n$ have radius at most $r_k
  + \tau$ with probability $1 - \mathcal{O}(n^{-c})$.
\end{lemma}
\begin{proof}
  To prove this lemma, it suffices to show that there exist constants
  $\kappa, \tau > 0$, such that the probability of a vertex with
  radius greater than $r_k + \tau$ having degree at least $k$, i.e.
  $\Pr[\exists v \in V \colon \deg(v) \ge k \land r(v) \ge r_k +
  \tau]$, is small.  To obtain the following sequence of inequalities,
  we first use the union bound, then apply the definition of
  conditional probabilities, and finally use
  Lemma~\ref{lem:smallerRadiusIncreasesNeighborhood}.
  \begin{align*}
    \Pr[\exists v \in V \colon \deg(v) \ge k \land r(v) \ge r_k + \tau] &\leq n \cdot \Pr[\deg(v) \ge k \land r(v) \ge r_k + \tau] \notag \\ 
                                                                        &\leq n \cdot \Pr[\deg(v) \ge k~\vert~r(v) \ge r_k + \tau]\\
                                                                        &\le n\cdot \Pr[\deg(v) \ge k~\vert~r(v) = r_k + \tau].
  \end{align*}
  To prove the statement of the lemma, it remains to show that
  $\Pr[\deg(v) \ge k \mid r(v) = r_k + \tau]$ is sufficiently small,
  i.e., in~$\mathcal O(n^{-(c+1)})$.

  Recall that, by Equation~(\ref{eq:expectedDegree}), the expected
  degree of a vertex with radius $r$ is in $\Theta(n e^{-r/2})$.  For
  a vertex $v$ with radius $r_k + \tau$ we obtain $ne^{-(r_k + \tau) /
    2} = e^{-\tau/2}k$.  It follows that there exists a constant $c' >
  0$, such that $\mathbb{E}[\deg(v)] \le c'e^{-\tau/2}k$.  By choosing
  $\tau$ large enough we can ensure that $k \ge 2e
  \mathbb{E}[\deg(v)]$, allowing us to apply the Chernoff-Hoeffding
  bound in Theorem~\ref{thm:chernoff-bound-original}.  We obtain
  $\Pr[\deg(v) \ge k] \le 2^{-k}$.  Finally, since $k \ge \kappa \log
  n$, we can choose $\kappa$ such that this probability is bounded by
  $\mathcal{O}(n^{-(c+1)})$.
\end{proof}

We are now ready to bound the number of vertices in a sector that have
degree at least~$k$.  As mentioned earlier, this bound only works for
large $k$ as the degree is not sufficiently concentrated otherwise.
Moreover, the degree cannot be too large, as otherwise the number of
vertices of this degree is not concentrated.  The upper bound on $k$
in the following lemma directly corresponds to our choice for
$\rho_I(\varphi)$.  Additionally, $\rho_O$ is chosen such that the
degrees of vertices with radii smaller than $\rho_O$ meets the lower
bound on $k$, i.e., the lemma holds for the central
part~$S|_{\rho_I(\varphi)}^{\rho_O}$.

\begin{lemma}
  \label{lem:numberVerticesWithDegreeLeastK}
  Let $G$ be a hyperbolic random graph and let $S$ be a sector of
  angular width $\varphi$.  If $k = \omega(\log n)$ and $k =
  \mathcal{O}((\varphi n / \log n)^{1/(2\alpha)})$, then the number of
  vertices in $S$ with degree at least $k$ is in $\mathcal O(\varphi n
  k^{-2\alpha})$ with probability $1 - \mathcal{O}(n^{-c})$ for any
  constant $c > 0$.
\end{lemma}
\begin{proof}
  By Lemma~\ref{lem:probDegreeTooLarge} we know that, for any constant
  $c' > 0$, there are constants $\kappa, \tau > 0$ such that all
  vertices of degree at least $k \ge \kappa \log n$ have radius at
  most $r_k + \tau$, with probability $1 - \mathcal{O}(n^{-c'})$.
  Since $k = \omega(\log n)$ we have $k \ge \kappa \log n$ for large
  enough $n$ and obtain that, with the same probability, all vertices
  of degree at least $k$ that are in $S$ are in $S|_0^{r_k + \tau}$.
  Since the angular width of $S$ is $\varphi$ and since the angular
  coordinates of the vertices are distributed uniformly, the expected
  number of vertices in $S|_0^{r_k + \tau}$ is given by
  $\varphi/(2\pi) \cdot n \mu(D_R|_0^{r_k + \tau})$.  Now we can apply
  Equation~(\ref{eq:originBallApproximation}), which states that a
  disk of radius $r_k + \tau$ centered at the origin has measure
  $e^{-\alpha(R - (r_k + \tau))}(1 + o(1))$ and obtain
  \begin{align*}
    \mathbb{E}[|\{v \in S|_0^{r_k + \tau}\}|] &= \frac{\varphi}{2\pi} n \mu(D_R|_0^{r_k + \tau})\\
                                              &= \frac{\varphi}{2\pi} ne^{-\alpha(R - (r_k + \tau))}(1 + o(1))\\
                                              &= \frac{\varphi}{2\pi} ne^{-2\alpha \log k - \alpha (C - \tau)} (1 + o(1)) \\
                                              &= \Theta(\varphi n k^{-2\alpha}).
  \end{align*}
  Note that $k = \mathcal{O}((\varphi n / \log n)^{1/(2\alpha)})$
  (which is a precondition of this lemma) implies that $\varphi n
  k^{-2\alpha} = \Omega(\log n)$.  Thus, we can apply the
  Chernoff-Hoeffding bound in Corollary~\ref{col:chernoff-bound} to
  conclude that $|\{v \in S|_0^{r_k + \tau}\}| = \mathcal{O}(\varphi n
  k^{-2\alpha})$ holds with probability $1 - \mathcal{O}(n^{-c})$ for
  any constant $c > 0$.
\end{proof}

Using these results, we can now bound the size of the search space in
the central part~$S|_{\rho_I(\varphi)}^{\rho_O}$ of our sector $S$,
yielding the following lemma.  (We note that the lower bound on
$\varphi$ that is a requirement of the following lemma, is weaker than
the one we need for Theorem~\ref{thm:main-theorem-first-phase}.)

\begin{lemma}
  \label{lem:sumOfDegreesCentralPart}
  Let $G$ be a hyperbolic random graph.  For every sector $S$ of
  angular width~$\varphi \in \Omega(\log(n)^{2\alpha + 1} / n)$, the
  degrees of the vertices in $S|_{\rho_I(\varphi)}^{\rho_O}$ sum to
  $\mathcal{O}(\varphi n)$ with high probability.
\end{lemma}
\begin{proof}
  First note that, analogous to the argumentation about sectors in the
  inner part of the disk, we can choose at most $\mathcal{O}(n)$
  sectors of width $2\varphi$ such that any sector of width $\varphi$
  lies completely in one of them.  Thus, the probability that there
  exists a sector where the sum of the vertex degrees in the central
  part of the disk is too large, is bounded by the probability that it
  is too large in at least one of these $\mathcal{O}(n)$ sectors (of
  twice the width).  In the following, we show for a single sector $S$
  of angular width $\varphi' = 2\varphi$ that the probability that the
  sum is too large is $\mathcal O(n^{-2})$.  The union bound then
  yields the claim, that the bound holds for every sector of angular
  width $\varphi$.

  To sum the degrees of all vertices in $S$, think of a vertex $v$ of
  degree $\deg(v)$ as a rectangle of height~$1$ and width $\deg(v)$.
  For a small graph, Figure~\ref{fig:degreeSumToIntegral} shows all
  such rectangles stacked on top of each other, sorted by their
  degree.  Note that the sum of degrees is equal to the area under the
  function $g(x) = |V^S_x|$ where $V^S_x = \{ v \in S~\vert~ \deg(v)
  \geq x\}$ is the set of vertices in $S$ that have degree at least
  $x$.
  Also note that the above considerations do not take into account
  that we sum only the degrees of vertices in the central part
  $S|_{\rho_I(\varphi')}^{\rho_O}$ of~$S$.  To resolve this,
  let~$k_{\min}$ and $k_{\max}$ be the minimum and maximum degree of
  vertices in~$S|_{\rho_I(\varphi')}^{\rho_O}$, respectively.  One can
  see in Figure~\ref{fig:degreeSumToIntegral} that summing only those
  degrees that are larger than~$k_{\min}$ is equivalent to integrating
  over $|V^S_{\max(k_{\min}, x)}|$ instead of $|V^S_x|$.  Thus, we can
  compute the sum of all degrees as
  \begin{align*}
    \sum_{v \in S|_{\rho_I(\varphi')}^{\rho_O}} \deg(v) &\le \sum_{\substack{v \in S,\\k_{\min} \leq \deg(v) \leq k_{\max}}} \deg(v) \\
                                                        &= \int_{0}^{k_{\max}} \vert V_{\max(k_{\min}, x)}^{S} \vert \dif x \notag \\
                                                        &= k_{\min} \vert V_{k_{\min}}^{S} \vert + \int_{k_{\min}}^{k_{\max}} \vert V_{x}^{S} \vert \dif x.
  \end{align*}

  \begin{figure}[t]
    \centering
    \includegraphics{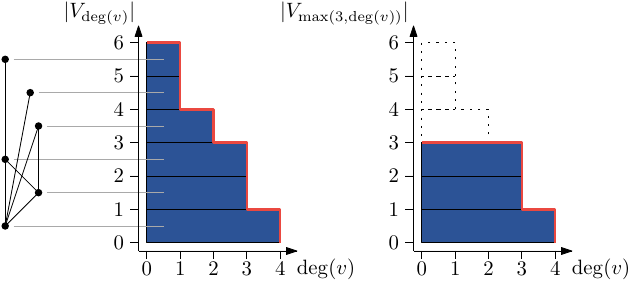}
    \caption{Visualization of how the sum over the degrees can be
      turned into an integral (left).  The same visualization but only
      the degrees of vertices with degree at least $3$ are summed up
      (right).}
    \label{fig:degreeSumToIntegral}
  \end{figure}

  To compute this integral, we first calculate the minimum and maximum
  degrees $k_{\min}$ and $k_{\max}$.  Afterwards, we apply
  Lemma~\ref{lem:numberVerticesWithDegreeLeastK} to bound $|V_x^S|$.
  For the minimum degree $k_{\min}$, assume that vertex $v$ has radius
  $\rho_O = R - (2+\varepsilon)\log\log(n)$ for any constant
  $\varepsilon \in (0, 1)$.  Using Equation~\eqref{eq:expectedDegree}
  the expected degree of $v$ is
  \begin{align*}
    \mathbb{E}[\deg(v)] = \Theta(n e^{-R/2 + (1 + \varepsilon / 2)\log\log(n)}) = \Theta(\log(n)^{1 + \varepsilon / 2}).
  \end{align*}
  Since $\varepsilon > 0$, this bound is $\omega(\log n)$, allowing us
  to apply the Chernoff-Hoeffding bounds in
  Corollaries~\ref{col:chernoff-bound}
  and~\ref{col:chernoff-lower-bound} to conclude that
  $\deg(v) = \Theta(\log(n)^{1 + \varepsilon / 2})$ with high
  probability.  Note that this only holds under the assumption that
  $v$ has radius exactly $\rho_O$.  However, by
  Lemma~\ref{lem:smallerRadiusIncreasesNeighborhood} all vertices with
  smaller radius have larger expected degree.  Therefore,
  $\Theta(\log(n)^{1 + \varepsilon / 2})$ is a lower bound on the
  expected degrees of all such vertices, allowing us to apply
  Corollary~\ref{col:chernoff-lower-bound} together with a union
  bound, to conclude that, with high probability, no vertex with
  smaller radius has smaller degree.  Thus, with high probability, the
  minimum degree in~$S|_{\rho_I(\varphi')}^{\rho_O}$ is
  $k_{\min} = \Theta(\log(n)^{1 + \varepsilon / 2})$.  Analogously,
  the bound on the maximum degree $k_{\max}$ of a vertex
  in~$S|_{\rho_I(\varphi')}^{\rho_O}$ can be obtained as follows.  Let
  $v$ be a vertex with radius
  $\rho_I(\varphi') = R - 1/\alpha(\log(\varphi' / 2\pi) + \log n -
  \log\log n)$.  The expected degree of $v$ is
  $\mathbb{E}[\deg(v)] = \Theta((\varphi' n / \log n)^{1/(2 \alpha)})$
  (Equation~\eqref{eq:expectedDegree}). Since
  $\varphi' = 2\varphi \in \Omega(\log(n)^{2\alpha + 1}/n)$, which is
  a precondition of this lemma, we can conclude that this bound on the
  expected degree of $v$ is $\Omega(\log n)$, allowing us to apply
  Corollary~\ref{col:chernoff-bound} to conclude that
  $\mathbb{E}[\deg(v)] = \mathcal{O}((\varphi n / \log
  n)^{1/(2\alpha)})$ holds with high probability.  Again, this only
  holds under the assumption that $v$ has radius exactly
  $\rho_I(\varphi')$.  However, by
  Lemma~\ref{lem:smallerRadiusIncreasesNeighborhood} all vertices with
  larger radius have smaller expected degree.  Therefore,
  $\mathcal{O}((\varphi n / \log n)^{1/(2\alpha)})$ is a valid upper
  bound on all their expected degrees, allowing us to apply
  Corollary~\ref{col:chernoff-bound} together with a union bound, to
  conclude that no vertex with larger radius has larger degree.  Thus,
  the maximum degree in $S|_{\rho_I(\varphi')}^{\rho_O}$ is
  $k_{\max} = \mathcal{O}((\varphi n / \log n)^{1/(2\alpha)})$ with
  high probability.

  Using Lemma~\ref{lem:numberVerticesWithDegreeLeastK} we obtain
  $|V_{x}^{S}| = \mathcal{O}(\varphi n x^{-2\alpha})$ with probability
  $1 - \mathcal{O}(n^{-c})$ for any constant $c > 0$.  Note that the
  requirements $x = \omega(\log n)$ and $x = \mathcal{O}((\varphi n /
  \log n)^{1/(2\alpha)})$ in
  Lemma~\ref{lem:numberVerticesWithDegreeLeastK} are satisfied as
  $k_{\min} \le x \le k_{\max}$.  By choosing $c = 2$ and applying the
  union bound over all degrees, we can conclude that, with high
  probability
  \begin{align}
    \sum_{v \in S|_{\rho_I(\varphi')}^{\rho_O}} \deg(v) &= \mathcal O(\varphi n k_{\min}^{-(2\alpha - 1)}) + \mathcal O(\varphi n \cdot \int_{k_{\min}}^{k_{\max}} x^{-2\alpha} \dif x) \notag \\
                                                        &= \mathcal O(\varphi n k_{\min}^{-(2\alpha - 1)}) + \mathcal O(\varphi n \cdot k_{\min}^{-(2\alpha - 1)}(1 - (k_{\min}/k_{\max})^{2\alpha - 1})). \notag
  \end{align}
  As $k_{\min} \le k_{\max}$, this can be further simplified to
  $\mathcal O(\varphi n k_{\min}^{-(2\alpha - 1)})$, which is
  $\mathcal O(\varphi n)$ since $k_{\min} = \omega(\log n)$.
\end{proof}

It remains to bound the sum of the degrees of vertices in the central
part of the disk~$D_R|_{\rho_I(2\pi)}^{\rho_O}$ that lie in the
neighborhood of a vertex $v$ with radius $\rho$, i.e., vertices lying
in $D_R(v)$.  Similar to the bounds for a sector $S$, we bound the sum
of degrees in $D_R(v)$ by bounding the number of vertices with a fixed
degree $k$ for every possible value of~$k$.  If all these bounds hold
with probability $1-\mathcal{O}(n^{-3})$, then the union bound shows
that the sum is concentrated with probability $1-\mathcal{O}(n^{-2})$.
To obtain a bound that holds for every possible angular coordinate of
$v$ (as claimed in Section~\ref{search-space-in-second-phase}), we
apply Lemma~\ref{lem:arbitrary-angle}.  There, we choose the random
variables $X_w$ to represent the degrees of the vertices.  Our bound
on the sum that holds with probability $1 - \mathcal{O}(n^{-2})$ at a
fixed angular coordinate, can then be translated to the same
asymptotic bound that holds with probability $1 - \mathcal{O}(n^{-1})$
at every possible angular coordinate.

\begin{figure}[t]
  \centering
  \includegraphics{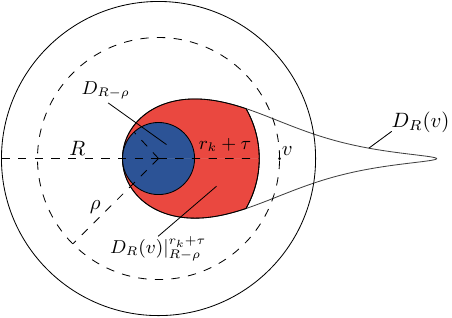}
  \caption{Determining the sum of degrees of the neighbors of vertex
    $v$ that are all contained in~$D_R(v)$.  To compute the measure of
    $D_R(v)|_{0}^{r_k + \tau}$ we divide it into two regions
    $D_R(v)|_0^{R - \rho} = D_{R - \rho}$ (blue) and $D_R(v)|_{R -
      \rho}^{r_k + \tau}$ (red).}
  \label{fig:secondPhaseDegreeSum}
\end{figure}

For a fixed degree $k = \omega(\log n)$, all vertices with degree at
least $k$ have radius at most $r_k + \tau$ with high probability due
to Lemma~\ref{lem:probDegreeTooLarge}, where $r_k = 2\log(n/k)$ and
$\tau$ is constant.  Thus, all vertices of degree at least $k$ in
$D_R(v)$ lie in $D_R(v)|_0^{r_k+\tau}$, with high probability.  In
analogy to Lemma~\ref{lem:numberVerticesWithDegreeLeastK}, we obtain
the following bound on the number of vertices in
$D_R(v)|_0^{r_k+\tau}$.

\begin{lemma}
  \label{lem:numberOfVerticesWithCertainDegreeNeighborhood}
  Let $G$ be a hyperbolic random graph and let $v$ be a vertex with
  radius $\rho = 1/\alpha (\log n - \log\log n)$.  If
  $k = \omega(\log n)$, the number of neighbors of $v$ with degree at
  least~$k$ is
  \begin{align*}
    |\{w \in N(v) \mid \deg(w) \ge k\}| = \mathcal{O}(n^{1-1/(2\alpha)} \log(n)^{1/(2\alpha)}k^{-(2\alpha - 1)} + \log n) 
  \end{align*}
   with probability $1-\mathcal O(n^{-c})$ for any constant $c
  > 0$.
\end{lemma}
\begin{proof}
  Since $k = \omega(\log n)$, we can apply
  Lemma~\ref{lem:probDegreeTooLarge} stating that all vertices of
  degree at least~$k$ in $D_R(v)$ lie within $D_R(v)|_{0}^{r_k+\tau}$
  with high probability.  To bound the number of neighbors of $v$ with
  degree at least $k$ we first compute the measure
  $\mu(D_R(v)|_0^{r_k+\tau})$.  To do this, we separate
  $D_R(v)|_0^{r_k+\tau}$ into the disk
  $D_R(v)|_0^{R-\rho} = D_{R - \rho}$ and
  $D_R(v)|_{R-\rho}^{r_k+\tau}$; see
  Figure~\ref{fig:secondPhaseDegreeSum}.  Due to
  Equation~\eqref{eq:originBallApproximation}, we have
  $\mu(D_{R-\rho}) = \mathcal O(e^{-\alpha(R - (R - \rho))}) =
  \mathcal O(\log n/n)$, which is already an upper bound on
  $\mu(D_R(v)|_{0}^{r_k + \tau})$ for the case where
  $r_k + \tau \le R - \rho$.  When $r_k + \tau > R - \rho$, we need to
  add the measure of $D_R(v)|_{R-\rho}^{r_k+\tau}$, which is given by
  \begin{align*}
    \mu(D_R(v)|_{R-\rho}^{r_k+\tau}) = \int_{R - \rho}^{r_k+\tau} 2 \int_{0}^{\theta(\rho, r)} f(r, \phi) \dif \phi \dif r = \mathcal O \left( \int_{R - \rho}^{r_k+\tau} \theta(\rho, r) f(r) \dif r \right)
  \end{align*}
  Since we consider $r \in [R - \rho, r_k+\tau]$ in the integral, we
  have $r \ge R - \rho$, allowing us to apply
  Equation~\eqref{eq:maxAdjacentAngle} to conclude that $\theta(\rho,
  r) = \mathcal O(e^{(R - \rho - r)/2})$.  Furthermore, we can
  substitute the probability density $f(r) = \mathcal O(e^{-\alpha(R -
    r)})$ (Equation~\eqref{eq:probDensity}) to obtain
  \begin{align*}
    \mu(D_R(v)|_{R-\rho}^{r_k+\tau}) &= \mathcal O \left( \int_{R - \rho}^{r_k+\tau} e^{(R - \rho - r)/2}\cdot e^{-\alpha(R - r)} \dif r \right)\\
                                     &= \mathcal O \left( e^{(R - \rho)/2}\cdot e^{-\alpha R} \cdot\int_{R - \rho}^{r_k+\tau} e^{(\alpha - 1/2) r} \dif r \right)\\
                                     &= \mathcal O \left(e^{-(\alpha - 1/2)R} \cdot e^{ - \rho/2} \cdot \left[ e^{(\alpha - 1/2)(r_k+\tau)} - e^{(\alpha - 1/2)(R - \rho)} \right] \right).
  \end{align*}
  Dropping the negative term in the brackets and substituting $R =
  2\log n + C$, $\rho = 1/\alpha (\log n - \log\log n)$, and $r_k =
  2\log(n/k)$, we obtain
  \begin{align*}
    \mu(D_R(v)|_{R-\rho}^{r_k+\tau}) &= \mathcal O \left(e^{-(\alpha - 1/2)R} \cdot e^{ - \rho/2} \cdot e^{(\alpha - 1/2)(r_k+\tau)} \right)\\
                                     &= \mathcal{O} \left(n^{-(2\alpha - 1)} \cdot n^{-1/(2\alpha)} \log(n)^{1/(2\alpha)} \cdot n^{2\alpha - 1}\cdot k^{-(2\alpha - 1)} \right)\\
                                     &= \mathcal{O} \left((\log(n) / n)^{1/(2\alpha)} \cdot k^{-(2\alpha - 1)} \right).
  \end{align*}

  The expected number of vertices in $D_R(v)|_0^{r_k + \tau}$ is now
  obtained by reversing the previous split and adding the measures of
  $D_{R - \rho}$ and $D_R(v)|_{R - \rho}^{r_k + \tau}$, which yields
  \begin{align*}
    \mathbb{E}[|\{v \in D_R(v)|_{0}^{r_k + \tau}\}|] &= n \cdot \left( \mu(D_{R - \rho}) + \mu(D_R(v)|_{R - \rho}^{r_k + \tau}) \right) \\
                                                     &= \mathcal{O}(\log n + n^{1 - 1/(2\alpha)} \log(n)^{1/(2\alpha)} k^{-(2\alpha - 1)})
  \end{align*}
  and it remains to show that this bound holds with large enough
  probability.  Clearly, this bound is at least logarithmic.  Thus, we
  can apply Corollary~\ref{col:chernoff-bound} to conclude that it
  holds with probability $1-\mathcal O(n^{-c})$ for any constant $c$.
\end{proof}

With this, we are now ready to bound the sum of the degrees of the
vertices in the central part of the disk that are in the neighborhood
of a vertex with radius $\rho$.  The proof of the following lemma is
analogous to the one of Lemma~\ref{lem:sumOfDegreesCentralPart}.

\begin{lemma}
  \label{lem:degrees-hypothetical-vertex}
  Let $G$ be a hyperbolic random graph and let $v$ be a hypothetical
  vertex with radius $\rho = 1/\alpha (\log n - \log\log n)$ and
  arbitrary angular coordinate.  The degrees of neighbors of $v$ in
  $D_R|_{\rho_I(2\pi)}^{\rho_O}$ sum to $\mathcal {\tilde
    O}(n^{1/(2\alpha)})$ with high probability.
\end{lemma}
\begin{proof}
  Recall that $D_R(v)$ is the disk containing all neighbors of $v$.
  To bound the sum of the degrees of the vertices
  in~$D_R(v)|_{\rho_I(2\pi)}^{\rho_O}$, we use basically the same
  proof as in Lemma~\ref{lem:sumOfDegreesCentralPart} except we use
  Lemma~\ref{lem:numberOfVerticesWithCertainDegreeNeighborhood}
  instead of Lemma~\ref{lem:numberVerticesWithDegreeLeastK}.  Thus,
  \begin{align*}
    \sum_{w \in D_R(v)|_{\rho_I(2\pi)}^{\rho_O}} \deg(w) &\le k_{\min} \vert V_{k_{\min}}^{D_R(v)} \vert + \int_{k_{\min}}^{k_{\max}} \vert V_{x}^{D_R(v)} \vert \dif x,
  \end{align*}
  where $V_{x}^{D_R(v)}$ is the set of vertices of degree at least $x$
  in $D_R(v)$ and $k_{\min}$ and $k_{\max}$ are the maximum and
  minimum degree in $D_R(v)|_{\rho_I(2\pi)}^{\rho_O}$, respectively.

  We start with computing $k_{\min}$ and $k_{\max}$.  Using
  Equation~\eqref{eq:expectedDegree} and
  Corollaries~\ref{col:chernoff-bound}
  and~\ref{col:chernoff-lower-bound}, we obtain that a vertex of
  radius $\rho_O = R - (2 + \varepsilon)\log\log n$, for any
  $\varepsilon \in (0, 1)$, has degree $k_{\min} = \Theta((\log n)^{1
    + \varepsilon / 2})$ with high probability.  Moreover, by the same
  argumentation as in the proof of
  Lemma~\ref{lem:sumOfDegreesCentralPart} no vertex with smaller
  radius has smaller degree, with high probability.  Additionally, a
  vertex with radius $\rho_I(2 \pi) = R - 1/\alpha(\log n - \log\log
  n)$ has degree $k_{\max} = \mathcal{O}((n / \log n)^{1/(2\alpha)})$
  and no vertex with larger radius has larger degree, with high
  probability.  It follows that we can use the bound shown in
  Lemma~\ref{lem:numberOfVerticesWithCertainDegreeNeighborhood} for
  $|V_{x}^{D_R(v)}|$.  Thus, we obtain
  \begin{align*}
    \sum_{w \in D_R(v)|_{\rho_I(2\pi)}^{\rho_O}} \deg(w) &= \mathcal{\tilde O}\big(k_{\min} \cdot n^{1-1/(2\alpha)}k_{\min}^{-(2\alpha - 1)}\big) + \mathcal {\tilde O}\big(n^{1-1/(2\alpha)} \int_{k_{\min}}^{k_{\max}} x^{-(2\alpha - 1)} \dif x \big).
  \end{align*}
  Replacing $k_{\min}$ and simplifying the first term in the sum
  yields $\mathcal {\tilde O}(n^{1-1/(2 \alpha)})$, which is smaller
  than the claimed bound.  For the second term, we obtain
  \begin{align*}
    \mathcal {\tilde O}\left(n^{1-1/(2\alpha)} \int_{k_{\min}}^{k_{\max}} x^{-(2\alpha - 1)} \dif x\right) &= \mathcal {\tilde O}\left(n^{1-1/(2\alpha)} \left[k_{\max}^{2-2\alpha} - k_{\min}^{2-2\alpha}\right]\right).
  \end{align*}
  Dropping the negative term and replacing $k_{\max} = \mathcal
  {\tilde O}(n^{1/(2\alpha)})$, we obtain $\mathcal {\tilde
    O}(n^{1-1/(2\alpha) + 1/\alpha - 1}) = \mathcal{\tilde
    O}(n^{1/(2\alpha)})$.
\end{proof}

\subsection{The Outer Part of the Disk}

At this point we have bounded the sum of the degrees of the vertices
with radius at most $\rho_O = R - (2 + \varepsilon)\log\log n$ (for
any constant $\varepsilon \in (0, 1)$) that lie in a sector of angular
width $\varphi \in \Omega(\log(n)^{2\alpha + 1}/n)$ or in the
neighborhood of a vertex with radius $\rho$.  It remains to bound the
sums when considering vertices with radii larger than $\rho_O$.

To bound the sum of the vertex degrees in the outer part of a sector
$S|_{\rho_O}^R$, we start by computing the expected value.

% This part contains many vertices, all of which have low expected
% degree.  To bound their sum with high probability, we consider the
% coordinates of the vertices as random variables and the sum of their
% degrees as a function in these variables.  Then, our plan to show
% concentration is to apply a method of average bounded
% differences~\cite[Theorem~7.2]{dp-cmara-12}.  It is based on the fact
% that changing the value of a single random variable (i.e., moving the
% position of a single vertex) has only little effect on the function
% (i.e., on the sum of degrees).  To make sure that this is actually
% true, we exclude certain bad events that happen only with low
% probability: First, the maximum degree in $S_{\rho_2}^R$ should not be
% too high such that moving a single vertex can increase its degree only
% slightly.  Second, there should not be too many vertices in
% $S_{\rho_2}^R$ so that the sum of degrees actually changes only for
% few vertices (as we do not count vertices not in $S_{\rho_2}^R$).
% Before we bound the sum of degrees in $S_{\rho_2}^R$ with high
% probability, we compute its expected value.  This is later used when
% applying the concentration bound.

\begin{lemma}
  \label{lem:expectationInOuterRingSector}
  Let $G$ be a hyperbolic random graph. For a sector $S$ of angular
  width~$\varphi$, the degrees of vertices in $S|_{\rho_O}^R$ sum to
  $\Theta(\varphi n)$ in expectation.
\end{lemma}
\begin{proof}
  Let $\deg(v)$ be the random variable describing the degree of a
  vertex $v$.  Moreover, let~$X_v$ be the indicator variable that is
  $1$ if $v \in S|_{\rho_O}^R$ and $0$ otherwise.  Then the expected
  sum of the degrees of vertices in $S|_{\rho_O}^R$ is given by
  \begin{align*}
    \mathbb{E}\left[ \sum_{v \in V} X_v \cdot \deg(v) \right] = \sum_{v \in V} \mathbb{E}[X_v\cdot\deg(v)] = n \cdot \Pr[v \in S|_{\rho_O}^R]\cdot\mathbb{E}[\deg(v)\mid v \in S|_{\rho_O}^R].
  \end{align*}

  Note that $\Pr[v\in S|_{\rho_O}^R]$ is simply the measure
  $\mu(S|_{\rho_O}^R)$.  As the angular coordinate is uniformly
  distributed, the whole sector $S$ has measure $\Theta(\varphi)$.
  Moreover, the region of the disk containing the points with constant
  distance to the boundary has constant measure.  Thus, the measure of
  $S|_{\rho_O}^R$ is also in $\Theta(\varphi)$.  For the sake of
  completeness, the measure of~$S|_{\rho_O}^R$ can be formally
  computed as
  \begin{align*}
    \mu(S|_{\rho_O}^R) &= \mu(S \setminus S|_0^{\rho_O}) \\
                       &= \frac{\varphi}{2\pi} (1 - \mu(D_{\rho_O})) \\
                       &= \frac{\varphi}{2\pi} (1 - e^{-\alpha(R - \rho_O)}(1 + o(1))) \notag \\
                       &= \frac{\varphi}{2\pi} (1 - \mathcal{O}((\log n)^{-\alpha(2 + \varepsilon)})) \notag \\
                       &= \Theta(\varphi).
  \end{align*}
  It remains to determine $\mathbb{E}[\deg(v)~\vert~v \in
  S|_{\rho_O}^R]$, which can be done as follows.
 \begin{align}
   \mathbb{E}[\deg(v) \mid v \in S|_{\rho_O}^R] &= \iint_{S|_{\rho_O}^R} \mathbb{E}[\deg(v)~\vert~r(v) = r] \frac{f(r, \phi)}{\mu(S|_{\rho_O}^R)} \dif \phi \dif r \notag \\
                                                &= \frac{1}{\mu(S|_{\rho_O}^R)}\cdot \int_{\rho_O}^{R} \int_{0}^{\varphi} \mathbb{E}[\deg(v)~\vert~r(v) = r] f(r, \phi) \dif \phi \dif r \notag \\
                                                &= \Theta(1) \cdot \int_{\rho_O}^{R} \mathbb{E}[\deg(v)~\vert~r(v) = r] f(r) \dif r \notag \\
                                                &= \Theta(1) \cdot n \cdot e^{-\alpha R} \int_{\rho_O}^{R} e^{(\alpha - 1/2)r} \dif r \notag \\
                                                &= \Theta(1) \cdot n \cdot e^{-\alpha R}\left[ e^{(\alpha - 1/2)R} -  e^{(\alpha - 1/2)\rho_O}  \right]\notag \\
                                                &= \Theta(1) \cdot n \cdot e^{-R/2} \left[ 1 - e^{-(\alpha - 1/2)(R - \rho_O)} \right] \notag
  \end{align}
  Note that the part in brackets is bounded by a constant.  Moreover,
  as $R = 2\log n + C$, $n\cdot e^{-R/2}$ is constant as well.  Thus,
  $\mathbb{E}[\deg(v) \mid v \in S|_{\rho_O}^R]$ is in $\Theta(1)$.
  It follows that the expected sum of the degrees is $\Theta(\varphi
  n)$.
\end{proof}

Unfortunately, the sum of the vertex degrees in $S|_{\rho_O}^{R}$ is
not concentrated sufficiently well around its expectation to conclude
that this bound also holds with high probability.  The problem lies
with the high-degree vertices in the graph, which can be adjacent to
none or all vertices in $S|_{\rho_O}^{R}$ depending on their
positions.  That is, small perturbations of the position of a single
high-degree vertex can change the sum by too much.  To overcome this
issue, we consider the impact of high-degree vertices separately.  To
this end, we partition the edge set that contributes to the degrees of
the vertices in $S|_{\rho_O}^{R}$ into two sets $E_I$ and $E_O$,
denoting the \emph{inner edges} where the other endpoint is in
$D_R|_0^{\rho_O}$ and the \emph{outer edges} where the other endpoint
is in $D_R|_{\rho_O}^{R}$.  The sum of the degrees of the vertices in
$S|_{\rho_O}^{R}$ can then be bounded by taking the number of inner
edges and adding them to twice the number of outer edges.  That is,
\begin{align*}
  \sum_{v \in S|_{\rho_O}^{R}} \deg(v) \le |E_I| + 2|E_O|.
\end{align*}

Since $E_I$ denotes all edges with one endpoint in $S|_{\rho_O}^{R}$
and the other in the inner or central part of the disk, we can obtain
an upper bound on the first summand by summing the degrees of the
vertices in $D_R|_0^{\rho_O}$ that are adjacent to any vertex in
$S|_{\rho_O}^{R}$.  Since $\rho \le \rho_O$, we have $S|_{\rho_O}^{R}
\subseteq S|_{\rho}^{R}$, allowing us to apply
Lemma~\ref{lem:boundaryCover} to conclude that all such vertices are
contained in $S$ or are neighbors of the two hypothetical corner
vertices $c_1$ and $c_2$, which both have radius $\rho$.  Thus,
$|E_I|$ can be bounded by the sum of the degrees of vertices in a
sector and in the neighborhood of a vertex with radius $\rho$, but
constrained to vertices in the inner and central parts of the disk.
Corresponding bounds that hold with high probability have been
determined above.  For the sector we obtain an upper bound of
$\mathcal{\tilde{O}}(\delta_{\max})$ for the inner part
(Corollary~\ref{cor:number_edges_in_inner_sector}) and
$\mathcal{O}(\varphi n)$ for the central part
(Lemma~\ref{lem:sumOfDegreesCentralPart}).  For the neighborhood of a
vertex with radius $\rho$ we have $\mathcal{\tilde{O}}(\delta_{\max})$
for the inner part
(Corollary~\ref{cor:number_edges_in_inner_neighborhood}) and
$\mathcal{\tilde{O}}(n^{1/(2\alpha)})$ for the central part
(Lemma~\ref{lem:degrees-hypothetical-vertex}).  Taking them together,
we obtain the following corollary.

\begin{corollary}
  \label{col:number-inner-edges}
  Let $G$ be a hyperbolic random graph.  For every sector $S$ of
  angular width $\varphi \in \Omega(\log(n)^{2\alpha + 1}/n)$, the
  number of edges with one endpoint in $S|_{\rho_O}^{R}$ and the other
  in $D_R|_0^{\rho_O}$ is in $\mathcal{\tilde{O}}(\varphi n +
  n^{1/(2\alpha)} + \delta_{\max})$, with high probability.
\end{corollary}

To obtain an upper bound on the second part of the above sum, we aim
to apply a \emph{method of typical bounded differences} based on the
fact that changing the position of a single vertex has
\emph{typically} only little impact on the number of outer edges.  The
idea is as follows.  We consider $|E_O|$ as a function that only
depends on the positions $P_1, \dots, P_n$ of the vertices in the
graph and we ask ourselves: How much can $|E_O|$ change, if we alter
the position of a single vertex $i$?  Clearly, this change can be
large in the worst case.  Assume that we move~$i$ from outside
$D_R|_{\rho_O}^{R}$ into $S|_{\rho_O}^{R}$.  Then, $i$ does not
contribute anything to $|E_O|$ before the move and the increase in
$|E_O|$ depends on the number of outer edges that are incident to~$i$
after the move, which can be $n - 1$ in the worst case.  However, it
is very unlikely that a vertex in $S|_{\rho_O}^{R}$ has this many
neighbors that lie in the outer part of the disk.  In fact, its degree
is typically much smaller.  To formalize this, we represent the
typical case using an event~$A$, denoting that the degree of such a
vertex is at most a constant factor larger than the expected degree of
a vertex with radius $\rho_O = R - (2 + \varepsilon)\log\log n$ for
any constant $\varepsilon \in (0, 1)$.  More precisely,~$A$ denotes
the event in which all disks of radius $R$ with center
in~$D_R|_{\rho_O}^{R}$ contain at most
$\mathcal{O}(\log(n)^{1+\varepsilon/2})$ vertices.  In this case,
moving a vertex $i$ in the same way as before leads to a much smaller
increase in the number of outer edges.  Assuming that~$A$ holds before
the move, there are at most $\mathcal{O}(\log(n)^{1 + \varepsilon/2})$
outer edges incident to $i$ after the move, which corresponds to the
increase of $|E_O|$.  The following lemma defines the event $A$
formally and shows that it holds with high probability.

\begin{lemma}
  \label{lem:typical-event-whp}
  Let $G$ be a hyperbolic random graph and let $\rho_O = R - (2 +
  \varepsilon)\log\log(n)$ for any constant $\varepsilon \in (0, 1)$.
  Then, all disks $D$ with radius $R$ and center in
  $D_R|_{\rho_O}^{R}$ contain at most $|\{v \in D\}| =
  \mathcal{O}(\log(n)^{1 + \varepsilon/2})$ vertices, with probability
  $1 - \mathcal{O}(n^{-c})$ for any constant $c$.
\end{lemma}
\begin{proof}
  Let $D$ be a disk of radius $R$ and center
  $P \in D_R|_{\rho_O}^{R}$.  By
  Lemma~\ref{lem:smallerRadiusIncreasesNeighborhood}, a valid upper
  bound on the expected number of vertices in $D$ can be obtained by
  considering the disk~$D'$ at center $P'$ instead, which has the same
  angular coordinate as~$P$ and radius $\rho_O$.  Thus, using
  Lemma~\eqref{eq:expectedDegree} we get
  \begin{align*}
    \mathbb{E}[|\{v \in D\}|] \le \mathbb{E}[|\{v \in D'\}|] = \mathcal{O}(n e^{-\rho_O / 2}) = \mathcal{O}(\log(n)^{1 + \varepsilon / 2}).
  \end{align*}
  Moreover, since $\varepsilon > 0$, this bound is $\omega(\log n)$
  and we can apply Corollary~\ref{col:chernoff-bound} to conclude that
  $|\{v \in D\}| \in \mathcal{O}(\log(n)^{1 + \varepsilon / 2})$ holds
  with probability $1 - \mathcal{O}(n^{-c'})$ for any constant $c'$.
  To obtain a bound that holds for every possible angular coordinate
  for $P$, we apply Lemma~\ref{fig:disk-cover}, which allows us to
  translate our bound that holds for any given disk $D$ with
  probability $1 - \mathcal{O}(n^{-c'})$ to the same asymptotic bound
  that holds with probability $1 - \mathcal{O}(n^{-c' + 1})$ for all
  possible angular coordinates.  Choosing $c' = c + 1$ then yields the
  claim.
\end{proof}

So while moving a single vertex leads to a large change in the number
of outer edges $|E_O|$ in the worst case, we observe only small
changes in the typical case $A$.  Formally, we say that a function $f
\colon \Omega^{n} \rightarrow \mathbb{R}$ satisfies the \emph{typical
  bounded differences condition} with respect to an event $A \subseteq
\Omega^{n}$ if for all $i \in \{1, \dots, n\}$ there exist
$\Delta_i^{A} \le \Delta_i$ such that
\begin{align*}
  \label{eq:typical-bounded-differences-condition}
  |f(\boldsymbol{x}) - f(\boldsymbol{x}')| \le
  \begin{cases}
    \Delta_i^{A}, & \text{if}~\boldsymbol{x} \in A,\\
    \Delta_i, & \text{otherwise},
  \end{cases}
\end{align*}
for all $\boldsymbol{x}, \boldsymbol{x}' \in \Omega^{n}$ that differ
only in the $i$th component.

\begin{theorem}[{Method of Typical Bounded Differences,~\cite[Theorem 2\protect\footnotemark]{w-mtbd-16}}]
  \footnotetext{We state a slightly simplified version in order to
    facilitate understandability.  The original theorem allows for the
    random variables $X_1, \dots, X_n$ to be defined in different
    sample spaces.}
  \label{thm:typical-bounded-differences}
  Let $X_1, \dots, X_n \in \Omega$ be independent random variables and
  let $A \subseteq \Omega^{n}$ be an event.  Furthermore, let $f
  \colon \Omega^{n} \rightarrow \mathbb{R}$ be a function that
  satisfies the typical bounded differences condition with respect to
  $A$ and with parameters $\Delta_i^{A} \le \Delta_i$ for $i \in \{1,
  \dots, n\}$.  Then for all $\varepsilon_1, \dots, \varepsilon_n \in
  (0, 1]$ there exists an event~$B$ satisfying $\bar{B} \subseteq A$
  and $\Pr[B] \le \Pr[\bar{A}] \cdot \sum_i 1/\varepsilon_i$, such
  that for $\Delta = \sum_i (\Delta_i^{A} + \varepsilon_i (\Delta_i -
  \Delta_i^{A}))^2$ and $t \ge 0$ it holds that
  \begin{align*}
    \Pr[f > \mathbb{E}[f] + t \land \bar{B}] \le e^{-t^2 / (2\Delta)}.
  \end{align*}
\end{theorem}

Intuitively, the choice of the values for $\varepsilon_i$ has two
effects.  On the one hand, choosing $\varepsilon_i$ small allows us to
compensate for a potentially large worst-case change $\Delta_i$.  On
the other hand, this also increases the bound on the probability of
the event $B$ that represents the atypical case.  However, in that
case one can still obtain meaningful bounds if the typical event $A$
occurs with high enough probability.  In the following, we show that
an upper bound on the expected value $\mathbb{E}[f]$ is sufficient to
apply the method of typical bounded differences, before applying it to
bound the number of outer edges in a sector.

\begin{corollary}
  \label{col:typical-bounded-differences}
  Let $X_1, \dots, X_n \in \Omega$ be independent random variables and
  let $A \subseteq \Omega^{n}$ be an event.  Furthermore, let $f
  \colon \Omega^{n} \rightarrow \mathbb{R}$ be a function that
  satisfies the typical bounded differences condition with respect to
  $A$ and with parameters $\Delta_i^{A} \le \Delta_i$ for $i \in \{1,
  \dots, n\}$ and let $g(n)$ be an upper bound on $\mathbb{E}[f]$.
  Then for all $\varepsilon_1, \dots, \varepsilon_n \in (0, 1]$,
  $\Delta = \sum_i (\Delta_i^{A} + \varepsilon_i (\Delta_i -
  \Delta_i^{A}))^2$, and $c \ge 1$ it holds that
  \begin{align*}
    \Pr[f > c g(n)] \le e^{-((c - 1) g(n))^2 / (2\Delta)} + \Pr[\bar{A}] \sum_i 1 / \varepsilon_i.
  \end{align*}
\end{corollary}
\begin{proof}
  Let $h(n) \ge 0$ be a function with $f' = f + h(n)$ such that
  $\mathbb{E}[f'] = g(n)$.  Note that $h(n)$ exists since $g(n) \ge
  \mathbb{E}[f]$.  As a consequence, we have $f \le f'$ for all
  outcomes of $X_1, \dots, X_n$ and it holds that
  \begin{align*}
    |f'(\boldsymbol{x}) - f'(\boldsymbol{x}')| = |f(\boldsymbol{x}) + h(n) - f(\boldsymbol{x}') - h(n)| = |f(\boldsymbol{x}) - f(\boldsymbol{x}')| ,
  \end{align*}
  for all $\boldsymbol{x}, \boldsymbol{x}' \in \Omega^{n}$.
  Consequently, $f'$ satisfies the typical bounded differences
  condition with respect to $A$ with the same parameters $\Delta_i^A
  \le \Delta_i$ as $f$.  Since $f \le f'$ it holds that
  \begin{align*}
    \Pr[f > c g(n)] \le \Pr[f' > c g(n)] = \Pr[f' > c \mathbb{E}[f']].
  \end{align*}
  By choosing $t = (c - 1)\mathbb{E}[f']$ this can be written as
  \begin{align*}
    \Pr[f' > c \mathbb{E}[f']] = \Pr[f' > \mathbb{E}[f'] + t].
  \end{align*}
  Theorem~\ref{thm:typical-bounded-differences} now guarantees the
  existence of an event $B$ with $\Pr[B] \le \Pr[\bar{A}] \cdot \sum_i
  1 / \varepsilon_i$ and $\bar{B} \subseteq A$, such that $\Pr[f' >
  \mathbb{E}[f'] + t \land \bar{B}] \le e^{-t^2/(2 \Delta)}$.  To
  bound $\Pr[f' > \mathbb{E}[f'] + t]$ we apply the law of total
  probability and consider the events $B$ and $\bar{B}$ separately
  \begin{align*}
    \Pr[f' > \mathbb{E}[f'] + t] = \Pr[f' > \mathbb{E}[f'] + t \mid \bar{B}] \cdot \Pr[\bar{B}] + \Pr[f' > \mathbb{E}[f'] + t \mid B] \cdot \Pr[B].
  \end{align*}
  The first part of the sum can be simplified using the definition of
  conditional probabilities.  Moreover, it holds that $\Pr[f' >
  \mathbb{E}[f'] + t \mid B] \le 1$.  Thus, we can bound the above
  term by
  \begin{align*}
    \Pr[f' > \mathbb{E}[f'] + t] \le \Pr[f' > \mathbb{E}[f'] + t \land \bar{B}] + \Pr[B].
  \end{align*}
  Both remaining summands can now be bounded using the upper bounds
  that we previously obtained by applying
  Theorem~\ref{thm:typical-bounded-differences}, i.e., $\Pr[f' >
  \mathbb{E}[f'] + t \land \bar{B}] \le e^{- t^2/(2 \Delta)}$ and
  $\Pr[B] \le \Pr[\bar{A}] \cdot \sum_i 1 / \varepsilon_i$.  Thus,
  \begin{align*}
    \Pr[f' > \mathbb{E}[f'] + t] \le e^{-t^2/(2 \Delta)} + \Pr[\bar{A}] \cdot \sum_i 1 / \varepsilon_i.
  \end{align*}
  Finally, since $t$ was chosen as $t = (c - 1)\mathbb{E}[f']$ and
  since $\mathbb{E}[f'] = g(n)$, we obtain the claimed bound.
\end{proof}

We are now ready to bound the number $|E_O|$ of outer edges, i.e,
edges that are incident to vertices in a sector $S|_{\rho_O}^{R}$ and
have their other endpoint in $D_R|_{\rho_O}^{R}$.

\begin{lemma}
  \label{col:number-outer-edges}
  Let $G$ be a hyperbolic random graph.  For every sector $S$ of
  angular width $\varphi \in \Omega(\log(n)^{2} / n^{1/2})$, the
  number of edges with one endpoint in $S|_{\rho_O}^{R}$ and the other
  in $D_R|_{\rho_O}^{R}$ is in $\mathcal{O}(\varphi n)$, with high
  probability.
\end{lemma}
\begin{proof}
  First note that, analogous to the proof of
  Lemma~\ref{lem:sumOfDegreesCentralPart}, we can cover the disk
  with~$\mathcal{O}(n)$ sectors of angular width $2\varphi$ such that
  any sector of angular width $\varphi$ lies completely in one of
  them.  In the following, we show that the claimed bound holds with
  probability $\mathcal{O}(n^{-2})$ for a single sector $S$ of twice
  the width.\footnote{We note that this factor of 2 vanishes in the
    asymptotics throughout the proof.}  Applying union bound then
  yields the claim.

  We consider $|E_O|$, the number of edges with one endpoint in
  $S|_{\rho_O}^{R}$ and the other in $D_R|_{\rho_O}^{R}$, as a
  function hat only depends on the positions $P_1, \dots, P_n$ of the
  vertices in the graph.  To show that $|E_O|$ does not exceed an
  upper bound with high probability, we aim to apply the method of
  typical bounded differences
  (Corollary~\ref{col:typical-bounded-differences}).  We represent the
  typical case with an event $A$, denoting that all disks $D$ of
  radius $R$ and center in $D_R|_{\rho_O}^{R}$ contain at
  most~$\mathcal{O}(\log(n)^{1 + \varepsilon/2})$ vertices for any
  constant $\varepsilon \in (0, 1)$.  In order to determine the
  parameters $\Delta_i^{A} \le \Delta_i$ for $i \in \{1, \dots, n\}$
  with which $|E_O|$ fulfills the typical bounded differences
  condition with respect to $A$, we have to bound the maximum change
  in $|E_O|$ obtained by moving a single vertex.  As argued before,
  this change is at most $\Delta_i = n - 1$ for all
  $i \in \{1, \dots, n\}$ in the worst case.  To bound the
  $\Delta_i^{A}$, we start with a configuration of vertex coordinates
  in which the event $A$ holds.  In this case, it is easy to see that
  moving a single vertex $i$ changes~$|E_O|$ by at most
  $\Delta_i^{A} = \mathcal{O}(\log(n)^{1 + \varepsilon / 2})$ for all
  $i \in \{1, \dots, n\}$, since the degree of $i$ is at most this
  large after the move and so is the number of outer edges it
  contributes to $|E_O|$.

  We are now ready to apply the method of typical bounded differences
  (Corollary~\ref{col:typical-bounded-differences}).  For an upper
  bound $g(n)$ on~$|E_O|$, any constant $c > 1$, and all
  $\varepsilon_1, \dots, \varepsilon_n \in (0, 1]$ it states that
  \begin{align*}
    \Pr[|E_O| > c g(n)] \le e^{-((c - 1) g(n))^2 / (2 \Delta)} + \Pr[\bar A] \sum_i 1 / \varepsilon_i,
  \end{align*}
  where $\Delta = \sum_i (\Delta_i^{A} + \varepsilon_i (\Delta_i -
  \Delta_i^{A}))^2$.  First note that a valid upper bound on the
  expected number of outer edges incident to vertices in
  $S|_{\rho_O}^{R}$ is given by the expected sum of the degrees of
  these vertices.  Thus, by
  Lemma~\ref{lem:expectationInOuterRingSector} we can choose $g(n) =
  \Theta(\varphi n)$.  Moreover, by choosing $\varepsilon_i = 1/n$ for
  all $i \in \{1, \dots, n\}$ and since $\Delta_i = n - 1$ and
  $\Delta_i^{A} = \mathcal{O}(\log(n)^{1 + \varepsilon / 2})$ for all
  $i \in \{1, \dots, n\}$, we can compute $\Delta$ as
  \begin{align*}
    \Delta &= \sum_i (\Delta_i^{A} + \varepsilon_i (\Delta_i - \Delta_i^{A}))^2 \\
           &= \mathcal{O} \left( n \cdot \left( \log(n)^{1 + \varepsilon / 2} + 1/n (n - \log(n)^{1 + \varepsilon / 2}) \right)^2 \right) \\
           &= \mathcal{O} \left( n \cdot \left( \log(n)^{1 + \varepsilon / 2} + (1 - o(1)) \right)^2 \right) \\
           &= \mathcal{O} \left( n \cdot \log(n)^{2 + \varepsilon} \right)
  \end{align*}
  Consequently, the above probability can be bounded by
  \begin{align*}
    \Pr[|E_O| > c g(n)] &\le \exp \left(- \Theta \left( \frac{(\varphi n)^2}{n \log(n)^{2 + \varepsilon}} \right) \right) + \Pr[\bar A] \cdot n^2 \\
                        &\le \exp \left(- \Theta \left( \frac{\varphi^2 n}{\log(n)^{2 + \varepsilon}} \right) \right) + \Pr[\bar A] \cdot n^2.
  \end{align*}
  Since $\varphi \in \Omega(\log(n)^{2} / n^{1/2})$ is a precondition
  of this lemma and since $\varepsilon < 1$, we can conclude that the
  fraction is $\omega(\log n)$, which means that the first summand is
  $\mathcal{O}(n^{-c'})$ for any constant~$c'$.  Moreover, by
  Lemma~\ref{lem:typical-event-whp} event $A$ holds with probability
  $1 - \mathcal{O}(n^{-c'})$ for any constant $c'$.  Choosing $c' = 3$
  then yields the claim.
\end{proof}

\subsection{The Complete Disk}

Having obtained the required bounds for the inner, central, and outer
parts of the disk, we can now combine them to bound the sum of the
degrees in a sector and in the neighborhoods of the hypothetical
corner vertices.  We start with
Theorem~\ref{thm:main-theorem-first-phase}, which bounds the sum of
degrees in a sector.  To improve readability, we restate the theorem
here.

% To make sure that the Theorem number in the re-statement matches the
% number in the original statement, we briefly tell latex to assume to
% be in section 3 here.
\setcounter{section}{3}
\begin{backInTime}{main-theorem-first-phase}
  \begin{theorem}
    \mainTheoremFirstPhaseText
  \end{theorem}
\end{backInTime}
% Undo the hack and assume to be back in section 4 again.
\setcounter{section}{4}
\begin{proof}
  For the inner and central parts of every sector the sum of the
  vertex degrees is bounded by
  $\mathcal{\tilde O}(\varphi n + \delta_{\max})$ with high
  probability due to Corollary~\ref{cor:number_edges_in_inner_sector}
  and Lemma~\ref{lem:sumOfDegreesCentralPart}.  As argued above, the
  sum of the degrees of the remaining vertices, i.e., vertices with
  radius at least $\rho_O$, can be bounded by counting the number of
  inner edges and adding twice the number of outer edges.  Since
  $\varphi \in \Omega(\log(n)^{2} / n^{1 / 2})$, we can apply
  Corollary~\ref{col:number-inner-edges} and
  Corollary~\ref{col:number-outer-edges} to conclude that the
  corresponding sum is bounded by
  $\mathcal{\tilde O}(\varphi n + n^{1/(2 \alpha)} + \delta_{\max})$,
  with high probability.
\end{proof}

Lastly, it remains to bound the sum of the degrees of the neighbors of
the hypothetical corner vertices that were used to bound the size of
the search space in the second phase.  Again, for the sake of
readability, we restate the corresponding lemma here.

% To make sure that the Theorem number in the re-statement matches the
% number in the original statement, we briefly tell latex to assume to
% be in section 3 here.
\setcounter{section}{3}
\begin{backInTime}{main-theorem-second-phase}
  \begin{lemma}
    \mainTheoremSecondPhaseText
  \end{lemma}
\end{backInTime}
% Undo the hack and assume to be back in section 4 again.
\setcounter{section}{4}
\begin{proof}
  For the inner and central parts of the neighborhood of a vertex with
  radius $\rho$ and arbitrary angular coordinate the sum of the
  degrees is bounded by $\mathcal{\tilde{O}}(\delta_{\max} + n^{1/(2
    \alpha)})$ with high probability, due to
  Corollary~\ref{cor:number_edges_in_inner_neighborhood} and
  Lemma~\ref{lem:degrees-hypothetical-vertex}.  For the sum of the
  degrees in the outer part of the disk, note that all neighbors of
  radius at least $\rho$ have angular distance at most $\varphi =
  \mathcal{O}(n^{-(1/\alpha - 1)})$; see
  Section~\ref{sec:search-space-first}.  Thus, we can use
  Theorem~\ref{thm:main-theorem-first-phase} to conclude that claimed
  bound holds for the sum of their degrees.  Note that if $\varphi$ is
  too small to meet the requirements of
  Theorem~\ref{thm:main-theorem-first-phase}, we can choose $\varphi =
  \mathcal{\tilde O}(n^{-1/2})$ as a valid upper bound to conclude
  that the sum of degrees in the outer part of the neighborhood is in
  $\mathcal{\tilde{O}}(n^{1/2})$, which is
  $\mathcal{\tilde{O}}(n^{1/(2\alpha)})$ for $\alpha \in (1/2, 1)$.
\end{proof}

\section{Conclusion}
\label{sec:conclusion}

\begin{figure}
  \centering
  \includegraphics[scale=0.9425]{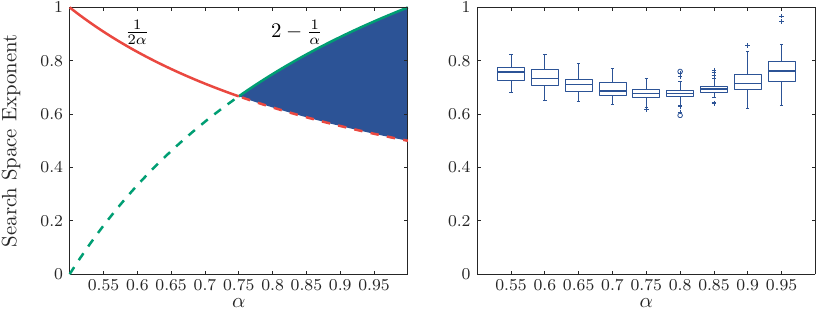}
  \caption{Left: The exponent of our theoretical bound depending on
    $\alpha$.  Right: The corresponding empirically measured search
    spaces.  The data was obtained by generating $20$ hyperbolic
    random graphs with average degree roughly $8$ for each shown
    $\alpha$ and each $n \in \{100\si{k}, 200\si{k}, 300\si{k}\}$.
    For each graph we sampled $300\si{k}$ start--destination pairs and
    report the maximum number of edges explored in one search.  The
    numbers are normalized with the total number of edges $m$ of the
    graph such that $x$ is plotted for a search space of size $m^x$.}
  \label{fig:realData}
\end{figure}
In the following, we briefly discuss why we think that the bound
$\mathcal{\tilde O}(n^{2-1/\alpha} \allowbreak + n^{1/(2\alpha)} +
\delta_{\max})$ is rather tight; see Figure~\ref{fig:realData}~(left)
for a plot of the exponents.  Clearly, the maximum degree of the graph
is a lower bound, i.e., we cannot improve the $\delta_{\max}$.  As
$\delta_{\max} = \tilde \Theta(n^{1/(2\alpha)})$ holds almost
surely~\cite{gpp-rhg-12}, we also cannot improve below
$\mathcal{\tilde O}(n^{1/(2\alpha)})$.  For the term $n^{2-1/\alpha}$
we do not have a lower bound.  Thus, the blue region in
Figure~\ref{fig:realData}~(left) is the only part where our bound can
potentially be improved.  However, by only making a single step from a
vertex with radius $\rho = 1/\alpha(\log n - \log\log n)$, we can
already reach vertices with angular distance $\Theta(n^{-(1/\alpha -
  1)})$.  Thus, it seems likely, that there exists a
start--destination pair such that all vertices within a sector of this
angular width are actually explored.  As such a sector contains
$\Theta(n^{2-1/\alpha})$ vertices, our bound seems rather tight (at
least asymptotically and up to poly-logarithmic factors).  For a
comparison of our theoretical bound with actual search-space sizes in
hyperbolic random graphs; see Figure~\ref{fig:realData}.

Finally, in order to put our results into perspective, we discuss the
following question: How does a heterogeneous degree distribution
impact the exponent in the running time of the bidirectional BFS?
First, considering networks with no underlying geometry, the exponent
is~$1/2$ for homogeneous networks and $(4 - \beta)/2 = 3/2 - \alpha$
for heterogeneous networks with power-law exponent
$\beta$~\cite{bn-kaabra-16}.  That is, when increasing the
heterogeneity by letting $\alpha$ go from~$1$ to $1/2$, the exponent
increases from $1/2$ to $1$.  This can be explained by the fact that a
heterogeneous degree distribution leads to high-degree vertices, which
leads to a higher running time when they are explored.

On hyperbolic random graphs, we get the same effect.  The
$1/(2\alpha)$-part of the exponent (the red function in
Figure~\ref{fig:realData}) is very similar to the above $3/2 -
\alpha$.  However, due to the underlying geometry, the heterogeneity
has another effect, expressed by the $2 - 1/\alpha$-part of the
exponent (the green function in Figure~\ref{fig:realData}).  This can
be explained as follows.  The underlying geometry constrains the parts
of the graph that a vertex can connect to.  As a result, the search
space cannot expand sufficiently fast on homogeneous networks and we
only get a constant speedup, i.e., the exponent is $1$. However,
increasing the heterogeneity leads to high degree vertices, which
accelerate the expansion of the search spaces, leading to a lower
exponent.

In conclusion, we can say that heterogeneity has two effects on the
bidirectional BFS:
\begin{enumerate}
\item More heterogeneity leads to higher running times as exploring
  high degree-vertices is costly.
\item More heterogeneity leads to lower running times as high
  degree-vertices let the search spaces expand quickly.
\end{enumerate}
For networks without underlying geometry, the second effect is
irrelevant, as the search space always expands quickly due to the
independence of edges.  Thus, the running time is better the more
homogeneous the network.  For networks with underlying geometry, both
effects play an important role leading to the v-shape in
Figure~\ref{fig:realData}.  For high heterogeneity ($\alpha < 0.75$),
the cost of exploring high degree vertices dominates, leading to the
exponent $1/(2\alpha)$.  For lower heterogeneity ($\alpha > 0.75$),
the slower expanding search space due to the underlying geometry
dominates, leading to the exponent $2 - 1/\alpha$.

\bibliography{efficient_shortest_paths}

\begin{thebibliography}{10}

\bibitem{ask-spqcn-12}
Takuya Akiba, Christian Sommer, and Ken-ichi Kawarabayashi.
\newblock Shortest-path queries for complex networks: Exploiting low tree-width
  outside the core.
\newblock In {\em Proceedings of the 15th International Conference on Extending
  Database Technology (EDBT 2012)}, page 144–155, 2012.
\newblock \href {https://doi.org/10.1145/2247596.2247614}
  {\path{doi:10.1145/2247596.2247614}}.

\bibitem{ama-m-16}
Gregorio Alanis-Lobato, Pablo Mier, and Miguel~A. Andrade-Navarro.
\newblock Manifold learning and maximum likelihood estimation for hyperbolic
  network embedding.
\newblock {\em Applied Network Science}, 1(10):1--14, 2016.
\newblock \href {https://doi.org/10.1007/s41109-016-0013-0}
  {\path{doi:10.1007/s41109-016-0013-0}}.

\bibitem{bffk-svcpt-20}
Thomas Bl{\"a}sius, Philipp Fischbeck, Tobias Friedrich, and Maximilian
  Katzmann.
\newblock {Solving Vertex Cover in Polynomial Time on Hyperbolic Random
  Graphs}.
\newblock In {\em 37th International Symposium on Theoretical Aspects of
  Computer Science (STACS 2020)}, pages 25:1--25:14, 2020.
\newblock \href {https://doi.org/10.4230/LIPIcs.STACS.2020.25}
  {\path{doi:10.4230/LIPIcs.STACS.2020.25}}.

\bibitem{bff-espsf-18}
Thomas Bl{\"a}sius, Cedric Freiberger, Tobias Friedrich, Maximilian Katzmann,
  Felix Montenegro-Retana, and Marianne Thieffry.
\newblock {Efficient Shortest Paths in Scale-Free Networks with Underlying
  Hyperbolic Geometry}.
\newblock In {\em 45th International Colloquium on Automata, Languages, and
  Programming (ICALP 2018)}, pages 20:1--20:14, 2018.
\newblock \href {https://doi.org/10.4230/LIPIcs.ICALP.2018.20}
  {\path{doi:10.4230/LIPIcs.ICALP.2018.20}}.

\bibitem{bfk-hrg-16}
Thomas Bl{\"a}sius, Tobias Friedrich, and Anton Krohmer.
\newblock Hyperbolic random graphs: Separators and treewidth.
\newblock In {\em 24th Annual European Symposium on Algorithms (ESA 2016)},
  pages 15:1--15:16, 2016.
\newblock \href {https://doi.org/10.4230/LIPIcs.ESA.2016.15}
  {\path{doi:10.4230/LIPIcs.ESA.2016.15}}.

\bibitem{bfk-chrg-17}
Thomas Bl{\"a}sius, Tobias Friedrich, and Anton Krohmer.
\newblock Cliques in hyperbolic random graphs.
\newblock {\em Algorithmica}, 80:2324--2344, 2018.
\newblock \href {https://doi.org/10.1007/s00453-017-0323-3}
  {\path{doi:10.1007/s00453-017-0323-3}}.

\bibitem{bfm-lchmcn-15}
Michel Bode, N.~Fountoulakis, and Tobias Müller.
\newblock On the largest component of a hyperbolic model of complex networks.
\newblock {\em Electronic Journal of Combinatorics}, 22:1--46, 2015.
\newblock \href {https://doi.org/10.1214/17-AAP1314}
  {\path{doi:10.1214/17-AAP1314}}.

\bibitem{bfm-gcrhg-13}
Michel Bode, Nikolaos Fountoulakis, and Tobias M{\"u}ller.
\newblock On the giant component of random hyperbolic graphs.
\newblock In {\em The Seventh European Conference on Combinatorics, Graph
  Theory and Applications (EUROCOMB 2013)}, pages 425--429, 2013.
\newblock \href {https://doi.org/10.37236/4958} {\path{doi:10.37236/4958}}.

\bibitem{bpk-sihm-10}
Mari{\'a}n Bogu{\~n}{\'a}, Fragkiskos Papadopoulos, and Dmitri Krioukov.
\newblock Sustaining the internet with hyperbolic mapping.
\newblock {\em Nature Communications}, 1:1--8, 2010.
\newblock \href {https://doi.org/10.1038/ncomms1063}
  {\path{doi:10.1038/ncomms1063}}.

\bibitem{bn-kaabra-16}
Michele Borassi and Emanuele Natale.
\newblock {KADABRA} is an adaptive algorithm for betweenness via random
  approximation.
\newblock In {\em 24th Annual European Symposium on Algorithms (ESA 2016)},
  pages 20:1--20:18, 2016.
\newblock \href {https://doi.org/10.4230/LIPIcs.ESA.2016.20}
  {\path{doi:10.4230/LIPIcs.ESA.2016.20}}.

\bibitem{bkl-adgcs-16}
Karl Bringmann, Ralph Keusch, and Johannes Lengler.
\newblock Average distance in a general class of scale-free networks with
  underlying geometry.
\newblock {\em CoRR}, abs/1602.05712, 2016.
\newblock URL: \url{http://arxiv.org/abs/1602.05712}.

\bibitem{bkl-sgirglt-17}
Karl Bringmann, Ralph Keusch, and Johannes Lengler.
\newblock Sampling geometric inhomogeneous random graphs in linear time.
\newblock In {\em 25th Annual European Symposium on Algorithms (ESA 2017)},
  pages 20:1--20:15, 2017.
\newblock \href {https://doi.org/10.4230/LIPIcs.ESA.2017.20}
  {\path{doi:10.4230/LIPIcs.ESA.2017.20}}.

\bibitem{bkl-graswp-17}
Karl Bringmann, Ralph Keusch, Johannes Lengler, Yannic Maus, and Anisur~Rahaman
  Molla.
\newblock Greedy routing and the algorithmic small-world phenomenon.
\newblock In {\em ACM Symposium on Principles of Distributed Computing (PODC
  2017)}, pages 371--380, 2017.
\newblock \href {https://doi.org/10.1145/3087801.3087829}
  {\path{doi:10.1145/3087801.3087829}}.

\bibitem{dp-cmara-12}
Devdatt~P. Dubhashi and Alessandro Panconesi.
\newblock {\em Concentration of Measure for the Analysis of Randomized
  Algorithms}.
\newblock Cambridge University Press, 2012.

\bibitem{fk-dhrg-18}
Tobias Friedrich and Anton Krohmer.
\newblock {On the Diameter of Hyperbolic Random Graphs}.
\newblock {\em SIAM Journal on Discrete Mathematics}, 32(2):1314--1334, 2018.
\newblock \href {https://doi.org/10.1137/17M1123961}
  {\path{doi:10.1137/17M1123961}}.

\bibitem{fss-dbtrg-13}
Tobias Friedrich, Thomas Sauerwald, and Alexandre Stauffer.
\newblock Diameter and broadcast time of random geometric graphs in arbitrary
  dimensions.
\newblock {\em Algorithmica}, 67:65–88, 2013.
\newblock \href {https://doi.org/10.1007/s00453-012-9710-y}
  {\path{doi:10.1007/s00453-012-9710-y}}.

\bibitem{gpp-rhg-12}
Luca Gugelmann, Konstantinos Panagiotou, and Ueli Peter.
\newblock Random hyperbolic graphs: Degree sequence and clustering.
\newblock In {\em 39th International Colloquium on Automata, Languages, and
  Programming (ICALP 2012)}, pages 573--585, 2012.
\newblock \href {https://doi.org/10.1007/978-3-642-31585-5_51}
  {\path{doi:10.1007/978-3-642-31585-5_51}}.

\bibitem{km-slcrhg-19}
Marcos Kiwi and Dieter Mitsche.
\newblock On the second largest component of random hyperbolic graphs.
\newblock {\em SIAM Journal on Discrete Mathematics}, 33(4):2200--2217, 2019.
\newblock \href {https://doi.org/10.1137/18M121201X}
  {\path{doi:10.1137/18M121201X}}.

\bibitem{kpk-h-10}
Dmitri Krioukov, Fragkiskos Papadopoulos, Maksim Kitsak, Amin Vahdat, and
  Mari\'an Bogu\~n\'a.
\newblock Hyperbolic geometry of complex networks.
\newblock {\em Physical Review E}, 82:036106, 2010.
\newblock \href {https://doi.org/10.1103/PhysRevE.82.036106}
  {\path{doi:10.1103/PhysRevE.82.036106}}.

\bibitem{lfc-dsspa-17}
Zheng Lu, Yunhe Feng, and Qing Cao.
\newblock Decentralized search for shortest path approximation in large-scale
  complex networks.
\newblock In {\em IEEE International Conference on Cloud Computing Technology
  and Science (CloudCom 2017)}, pages 130--137, 2017.
\newblock \href {https://doi.org/10.1109/CloudCom.2017.36}
  {\path{doi:10.1109/CloudCom.2017.36}}.

\bibitem{lr-bspagacb-89}
Michael Luby and Prabhakar Ragde.
\newblock A bidirectional shortest-path algorithm with good average-case
  behavior.
\newblock {\em Algorithmica}, 4(1):551–567, 1989.
\newblock \href {https://doi.org/10.1007/BF01553908}
  {\path{doi:10.1007/BF01553908}}.

\bibitem{ms-k-19}
Tobias Müller and Merlijn Staps.
\newblock The diameter of kpkvb random graphs.
\newblock {\em Advances in Applied Probability}, 51(2):358–377, 2019.
\newblock \href {https://doi.org/10.1017/apr.2019.23}
  {\path{doi:10.1017/apr.2019.23}}.

\bibitem{phzs-fafap-12}
Wei Peng, Xiaofeng Hu, Feng Zhao, and Jinshu Su.
\newblock A fast algorithm to find all-pairs shortest paths in complex
  networks.
\newblock {\em Procedia Computer Science}, 9:557 -- 566, 2012.
\newblock \href {https://doi.org/10.1016/j.procs.2012.04.060}
  {\path{doi:10.1016/j.procs.2012.04.060}}.

\bibitem{p-rgg-03}
Mathew Penrose.
\newblock {\em Random Geometric Graphs}.
\newblock Oxford University Press, 2003.

\bibitem{p-rgmcs-14}
Ueli Peter.
\newblock {\em Random Graph Models for Complex Systems}.
\newblock PhD thesis, ETH Zürich, 2014.

\bibitem{p-bhspp-69}
Ira~Sheldon Pohl.
\newblock {\em Bi-directional and Heuristic Search in Path Problems}.
\newblock PhD thesis, Stanford University, 1969.

\bibitem{w-mtbd-16}
Lutz Warnke.
\newblock On the method of typical bounded differences.
\newblock {\em Combinatorics, Probability and Computing}, 25(2):269–299,
  2016.
\newblock \href {https://doi.org/10.1017/S0963548315000103}
  {\path{doi:10.1017/S0963548315000103}}.

\end{thebibliography}
\end{document}